\documentclass[11pt,a4paper]{article}

\usepackage{amssymb,amsthm,amsmath,mathtools}
\usepackage{thmtools,thm-restate}
\usepackage[round]{natbib}
\usepackage{pifont}

\usepackage[usenames,svgnames,xcdraw,table]{xcolor}
\definecolor{DarkGreen}{rgb}{0.1,0.5,0.1}
\usepackage[backref=page]{hyperref}
\hypersetup{
	colorlinks=true,
	linkcolor=red,
	urlcolor=DarkGreen,
	citecolor=blue
}
\renewcommand*{\backref}[1]{}
\renewcommand*{\backrefalt}[4]{%
    \ifcase #1 (Not cited.)%
    \or        (Cited on page~#2)%
    \else      (Cited on pages~#2)%
    \fi}
\usepackage[capitalise,noabbrev]{cleveref}
\Crefname{property}{Property}{Properties}
\Crefname{example}{Example}{Examples}
\Crefname{table}{Table}{Tables}

\usepackage{cancel}
\usepackage{tabularx}
\usepackage{nicefrac}
\usepackage{enumerate}
\usepackage{etoolbox}
\usepackage[margin=0.99in]{geometry}
\usepackage[T1]{fontenc}
\usepackage[tt=false]{libertine}

\usepackage{url}
\usepackage[utf8]{inputenc}
\usepackage[small]{caption}
\usepackage{graphicx}
\usepackage{booktabs}
\usepackage{mathrsfs,amsfonts,dsfont,authblk}
\usepackage{array,multirow,graphicx,bigdelim}

\usepackage[linesnumbered,lined,boxed,ruled,vlined]{algorithm2e}

\SetKwInOut{Parameters}{Parameters}
\SetKwComment{Comment}{$\triangleright$\ }{}
\SetAlFnt{\small}
\SetAlCapFnt{\small}
\SetAlCapNameFnt{\small}
\SetAlCapHSkip{0pt}
\IncMargin{-\parindent}

\usepackage{tikz}
\usepackage{tikz-cd}
\usetikzlibrary{fit,calc,shapes,decorations.text,arrows,decorations.markings,decorations.pathmorphing,shapes.geometric,positioning,decorations.pathreplacing}
\tikzset{snake it/.style={decorate, decoration=snake}}

\colorlet{mygray}{gray!40}

\makeatletter
\let\oldnl\nl
\newcommand{\nonl}{\renewcommand{\nl}{\let\nl\oldnl}}
\makeatother

  {\list{}{\leftmargin=#1\rightmargin=#1}\item[]}%
  {\endlist}

\usepackage[labelfont={normalfont,bf},textfont=it]{caption}
\usepackage{subcaption}

\theoremstyle{definition}

\newtheorem{examplex}{Example}
\newenvironment{example}{\pushQED{\qed}\examplex}{\popQED\endexamplex}
\theoremstyle{remark}
\newtheorem{remark}{Remark}

\usepackage{flushend}
\usepackage[utf8]{inputenc}
\usepackage[inline]{enumitem}
\Crefname{claim}{Claim}{Claims}

\allowdisplaybreaks
\renewcommand{\>}{\succ}
\renewcommand{\nsucc}{\cancel{\>}}
\newcommand{\DA}{\texttt{DA}}

\renewcommand{\O}{\mathcal{O}}

\begin{document}

\title{Accomplice Manipulation of the Deferred Acceptance Algorithm}
\date{}

\author[1]{Hadi Hosseini}
\author[2]{Fatima Umar}
\author[3]{Rohit Vaish}
\affil[1]{Pennsylvania State University\\
	{\small\texttt{hadi@psu.edu}}}
\affil[2]{Rochester Institute of Technology\\
	{\small\texttt{fu1476@rit.edu}}}
\affil[3]{Tata Institute of Fundamental Research\\
	{\small\texttt{rohit.vaish@tifr.res.in}}}

\maketitle

\begin{abstract}
The deferred acceptance algorithm is an elegant solution to the stable matching problem that guarantees optimality and truthfulness for one side of the market. Despite these desirable guarantees, it is susceptible to strategic misreporting of preferences by the agents on the other side. We study a novel model of strategic behavior under the deferred acceptance algorithm: manipulation through an \emph{accomplice}. Here, an agent on the proposed-to side (say, a woman) partners with an agent on the proposing side---an accomplice---to manipulate on her behalf (possibly at the expense of worsening his match). We show that the optimal manipulation strategy for an accomplice comprises of promoting exactly one woman in his true list (i.e., an \emph{inconspicuous} manipulation). This structural result immediately gives a polynomial-time algorithm for computing an optimal accomplice manipulation. We also study the conditions under which the manipulated matching is stable with respect to the true preferences. Our experimental results show that accomplice manipulation outperforms self manipulation both in terms of the frequency of occurrence as well as the quality of matched partners.
\end{abstract}

\section{Introduction}\label{sec:Introduction}

The deferred acceptance (\DA{}) algorithm~\citep{GS62college} is a crowning achievement of the theory of two-sided matching, and forms the backbone of a wide array of real-world matching markets such as entry-level labor markets~\citep{R84evolution,RP99redesign} and school choice~\citep{APR+05boston,APR05new}. Under this algorithm, one side of the market (colloquially, the \emph{men}) makes proposals to the other side (the \emph{women}) subject to either immediate rejection or tentative acceptance. A key property of the \DA{} algorithm is \emph{stability} which says that no pair of unmatched agents should prefer each other over their assigned partners. This property has played a significant role in the long-term success of several real-world matching markets~\citep{R91natural,R02economist}.

The attractive stability guarantee of the \DA{} algorithm, however, comes at the cost of incentives, as any stable matching procedure is known to be vulnerable to strategic misreporting of preferences~\citep{R82economics}. The special proposal-rejection structure of the \DA{} algorithm makes truth-telling a dominant strategy for the proposing side, i.e., the men~\citep{DF81machiavelli,R82economics}, implying that any strategic behavior must occur on the proposed-to side, i.e., the women. This model of strategic behavior by a woman---which we call \emph{self manipulation}---has been the subject of extensive study in economics and computer science~\citep{DF81machiavelli,GS85some,DGS87further,TS01gale,KM09successful,KM10cheating,VG17manipulating,DST18coalitional}.

Our interest in this work is in studying a different model of strategic behavior under the \DA{} algorithm called \emph{manipulation through an accomplice}~\citep{BH19partners}. Under this model, a woman teams up with an agent on the proposing side (a.k.a. an accomplice) to manipulate the outcome on her behalf, possibly worsening his match in the process. Such a strategic alliance can naturally arise in the assignment of students to schools, where a ``well-connected'' student could have a school administrator manipulate on his/her behalf, possibly at a small loss to the school.

At first glance, manipulation through an accomplice might not seem any more useful than self manipulation, as the latter provides direct control over the preferences of the manipulator. Interestingly, there exist instances where this intuition turns out to be wrong.

\begin{example}[\textbf{Accomplice vs. self}]
Consider the following preference profile where the \DA{} outcome is underlined.

\begin{table}[h]
\centering 
    \begin{tabular}{p{0.019\textwidth}>{\centering}p{0.019\textwidth}>{\centering}p{0.019\textwidth}>{\centering}p{0.019\textwidth}>{\centering}p{0.019\textwidth}>{\centering}p{0.032\textwidth}>{\centering}p{0.019\textwidth}>{\centering}p{0.019\textwidth}>{\centering}p{0.019\textwidth}>{\centering}p{0.019\textwidth}>{\centering\arraybackslash}p{0.019\textwidth}}
         	$\boldsymbol{\textcolor{blue}{m_1}}$: & \underline{$w_3^{*}$} & $w_2$ & $w_1$ & $w_4$ && $\boldsymbol{\textcolor{blue}{w_1}}$: & $m_4$ & $m_3^{*}$ & $m_1$ & \underline{$m_2$}\\
            $m_2$: & \underline{$w_1$} & $w_4^{*}$ & $w_2$ & $w_3$ && $w_2$: & \underline{$m_4^{*}$} & $m_3$ & $m_2$ & $m_1$\\
            $m_3$: & $w_2$ & \underline{$w_4$} & $w_1^{*}$ & $w_3$ && $w_3$: & $m_3$ & \underline{$m_1^{*}$} & $m_2$ & $m_4$\\
            $m_4$: & \underline{$w_2^{*}$} & $w_1$ & $w_3$ & $w_4$ && $w_4$: & $m_2^{*}$ & $m_1$ & \underline{$m_3$} & $m_4$
    \end{tabular}%
\end{table}

Suppose $w_1$ seeks to improve her match via manipulation. The optimal self manipulation strategy for $w_1$ is truth-telling, as $m_2$ is the only man who proposes to her under the \DA{} algorithm. On the other hand, $w_1$ can improve her outcome by asking $m_1$ to misreport on her behalf. Indeed, if $m_1$ misreports by declaring $\>'_{m_1} \coloneqq w_1 \> w_3 \> w_2 \> w_4$, then $w_1$'s match improves from $m_2$ to $m_3$ (the new \DA{} matching is marked by $*$). Notice that the accomplice $m_1$ preserves his initial match, meaning he does not incur any `regret'.
\label{eg:Self_vs_Accomplice_Intro}
\end{example}

The above example highlights that accomplice manipulation could, in principle, have an advantage over self manipulation. However, it is not a priori clear how \emph{frequent} such an advantage might be in a typical matching scenario. To investigate the latter question, we take a quick experimental detour.

\paragraph{Accomplice manipulation is a viable strategic behavior.}
We simulate a two-sided matching scenario for an increasingly larger set of agents (specifically, $n \in \{3,\dots, 40\}$, where $n$ is the number of men/women) and for each setting, generate 1000 preference profiles uniformly at random. For each profile, we compute the optimal self manipulation under the \DA{} algorithm for a fixed woman~\citep{TS01gale}, as well as the optimal accomplice manipulation by any man (we allow any man to be chosen as an accomplice as long as he is not worse off, i.e., a \emph{no-regret} accomplice manipulation). \Cref{fig:AvTP_SvTP_AvS_SvA} illustrates the fraction of instances where accomplice and self manipulation are strictly more beneficial than truthful reporting, as well as how they compare against each other.

\begin{figure}[h]
    \centering
    \begin{minipage}{0.46\linewidth}
        \centering
        \includegraphics[width=\textwidth]{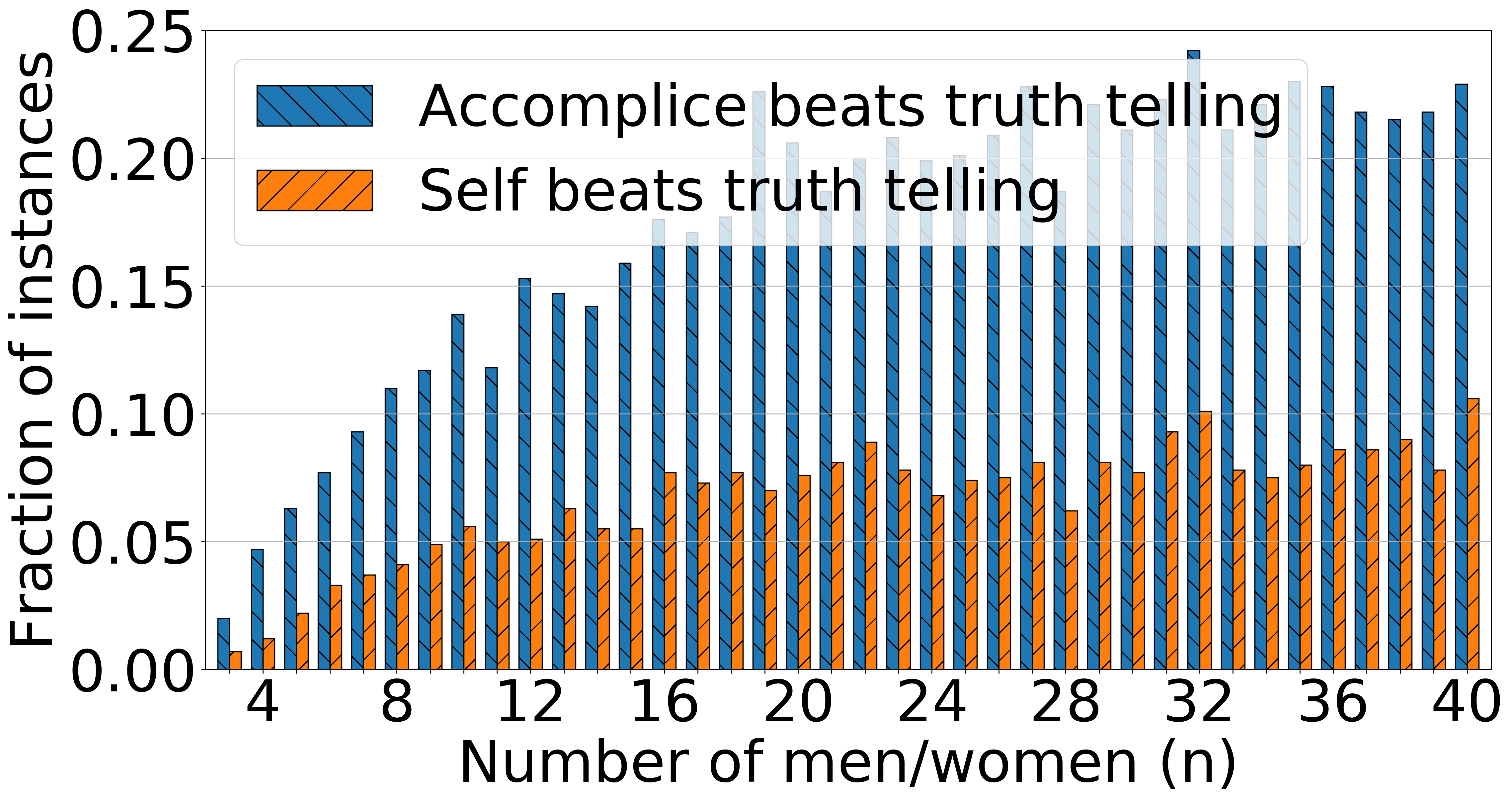}
    \end{minipage}
    \hfill
    \begin{minipage}{0.46\linewidth}
        \centering
        \includegraphics[width=\textwidth]{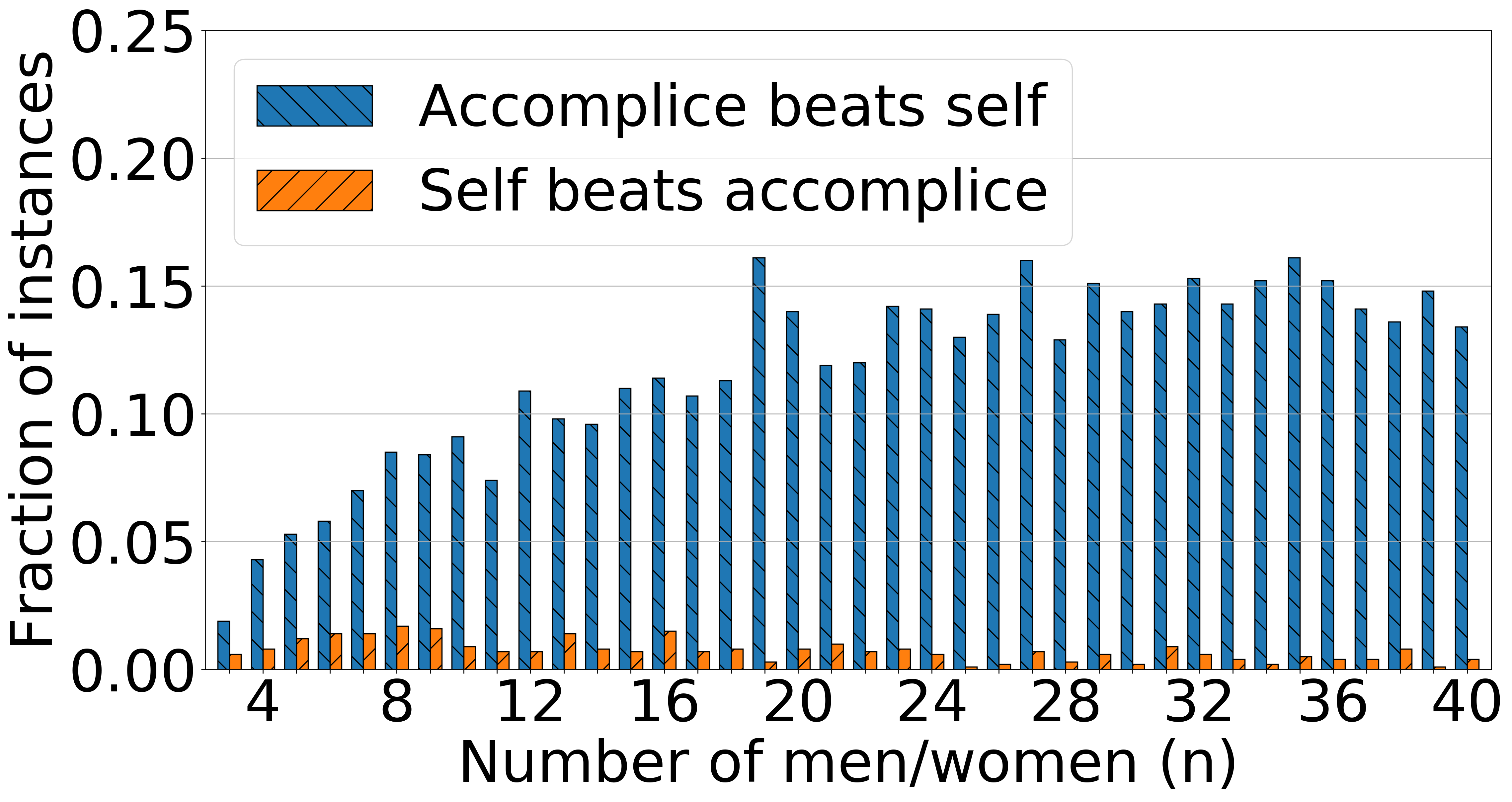}
    \end{minipage}
    \caption{Comparing no-regret accomplice manipulation and self manipulation against truthful reporting (left) and against each other (right).}
     \label{fig:AvTP_SvTP_AvS_SvA}
\end{figure}

From \Cref{eg:Self_vs_Accomplice_Intro} and \Cref{fig:AvTP_SvTP_AvS_SvA}, it is clear that the incentive for manipulation through an accomplice is not only present but actually more prevalent than self manipulation. Additionally, as we discuss later in our experimental results, women are expected to receive \emph{better} matches when manipulating through an accomplice (\Cref{fig:RankDiff_Accomplice_vs_Self}). These promising observations call for a systematic study of the structural and computational aspects of the accomplice manipulation problem, which is the focus of our work.

\paragraph{Our contributions.}
We consider two models of strategic behavior--- \emph{no-regret} manipulation (wherein the accomplice's match doesn't worsen upon misreporting) and \emph{with-regret} manipulation (where the accomplice could get a worse match)---and make the following contributions:

\begin{itemize}[leftmargin=*]
	\item \textbf{No-regret manipulation}: Our main theoretical result (\Cref{thm:Inconspicuous_NoRegret}) is that any optimal no-regret accomplice manipulation can be simulated by promoting exactly one woman in the true preference list of the accomplice; in other words, the manipulation is \emph{inconspicuous}~\citep{VG17manipulating}. This structural finding immediately gives a polynomial-time algorithm for computing an optimal manipulation (\Cref{cor:Optimal_No-regret_PolynomialTime}). We also show that the inconspicuous no-regret strategy results in a matching that is stable with respect to the true preferences (\Cref{cor:No-regret_Inconspicuous_Stable}). 
	
	\item \textbf{With-regret manipulation}: For the more permissible strategy space that allows the accomplice to incur regret, the optimal manipulation strategy once again turns out be inconspicuous (\Cref{thm:Inconspicuous_WithRegret}). However, in contrast to the no-regret case, the inconspicuous with-regret strategy is no longer guaranteed to be stability-preserving (\Cref{eg:WithRegret}). Nevertheless, any blocking pair can be shown to necessarily involve the accomplice (\Cref{prop:mStabilityExtension}). This property justifies the use of an accomplice who can be encouraged to tolerate some regret to benefit a woman.
	
	\item \textbf{Experiments}: On the experimental front, we work with preferences generated uniformly at random, and find that accomplice manipulation outperforms self manipulation with respect to the frequency of occurrence, the quality of matched partners, and the fraction of women who can improve their matches (\Cref{sec:Experiments}). 
\end{itemize}

\section{Related Work}

The impossibility result of \citet{R82economics} on the conflict between strategyproofness and stability has led to extensive follow-up research. Much of the earlier work in this direction focused on \emph{truncation} manipulation~\citep{GS85ms,R99truncation,CS14optimal,JKK14exhaustiveness}, where the misreported preference list is required to be a prefix of the true list. In the context of accomplice manipulation, however, the truncation manipulation problem becomes trivial. Indeed, for any fixed accomplice, it is easy to see that a truncation strategy is never better than truthful reporting as the set of proposals under the former is always a subset of those under the latter.

The literature on \emph{permutation} manipulation is more recent and has focused on computational aspects. \citet{TS01gale} provided a polynomial-time algorithm for computing the optimal permutation manipulation by a woman under the men-proposing \DA{} algorithm. \citet{VG17manipulating} showed that an optimal permutation manipulation is, without loss of generality, inconspicuous. They also studied conditions under which the manipulated outcome is stable with respect to the true preferences. \citet{DST18coalitional} generalized these results to permutation manipulation by a coalition of women.

Permutation manipulation has also been studied in the many-to-one setting (known as the \textit{college admissions} problem). In particular, \citet{A15susceptibility} have shown that under the student-proposing \DA{} algorithm, the problem of computing the optimal preference list for a college reduces to an analogous problem in the one-to-one setting, which, as noted above, can be solved in polynomial time~\citep{TS01gale}.

\citet{H06cheating} has studied (weakly) Pareto improving permutation manipulation by a coalition of men, revisiting the result of \citet{DF81machiavelli} on the impossibility of manipulations that are strictly improving for every member of the coalition.

Manipulation of the \DA{} algorithm has also been examined in a setting where the agents' preferences are drawn uniformly at random. \citet{K20instability} has shown that when one woman strategizes optimally, the outcome of the men-proposing \DA{} algorithm is ``close'' to the women-optimal matching (in terms of average rank of matched partners). The follow-up work of \citet{N20descending} studies how many independent strategic agents are needed (in expectation) in order to always achieve the women-optimal outcome.

The accomplice manipulation model was proposed by \citet{BH19partners}, who noted that manipulation through an accomplice can be strictly more preferable for the woman than optimal self manipulation. However, they left the structural and computational questions open.

\citet{HKN16improving} studied a closely related problem in school choice wherein the schools have the incentive to deliberately make themselves look less attractive to ``undesirable'' students (e.g., a private school that is legally required to cap its tuition fee for low-income students could make itself less attractive by increasing the rent in dormitories or requiring the students to purchase expensive uniforms). They showed that every stable matching mechanism is vulnerable to manipulation of quality scores by the schools, but such incentives vanish in a large-market setting.


\section{Preliminaries}
\label{sec:preliminaries}

\subsection{Stable Matching Problem}
\label{sec:Preliminaries_Stable_Matching}

\paragraph{Problem setup} An instance of the \emph{stable marriage problem}~\citep{GS62college} is specified by the tuple $\langle M, W, \> \rangle$, where $M$ is a set of $n$ men, $W$ is a set of $n$ women, and $\>$ is a \textit{preference profile} which consists of the preference lists of all agents. The preference list of any man $m \in M$, denoted by $\>_m$, is a strict total order over all women in $W$ (for any $w \in W$, the list $\>_w$ is defined analogously).

We use the shorthand $w_1 \succeq_m w_2$ to denote `either $w_1 \>_m w_2$ or $w_1 = w_2$', and write $\>_{-m}$ to denote the preference lists of all men and women except man $m$; thus, $\> = \{\>_{-m}, \>_m\}$.

\paragraph{Stable matchings} A \emph{matching} is a function $\mu: M \cup W \rightarrow M \cup W$ such that $\mu(m) \in W$ for all $m \in M$, $\mu(w) \in M$ for all $w \in W$, and $\mu(m) = w$ if and only if $\mu(w) = m$. A matching $\mu$ admits a \emph{blocking pair} with respect to the preference profile $\>$ if there is a man-woman pair $(m, w)$ who prefer each other over their assigned partners under $\mu$, i.e., $w \>_m \mu(m)$ and $m \>_w \mu(w)$. A \emph{stable matching} is one that does not admit any blocking pair. We will write $S_{\>}$ to denote the set of all matchings that are stable with respect to $\>$. In the worst case, the size of $S_{\>}$ grows exponentially with respect to $n$. \citet{K97stable} showed that the lower bound on the number of stable matchings is $2.28^n$, and \citet{KGW18simply} found the current best upper bound to be $c^n$ for a universal constant $c$. 

In addition, for any pair of matchings $\mu,\mu'$, we will write $\mu \succeq_M \mu'$ to denote $\mu(m) \succeq_m \mu'(m)$ for all $m \in M$ (and $\mu \succeq_W \mu'$ for the women).

\paragraph{Deferred acceptance algorithm} Given a preference profile $\>$, the Deferred Acceptance (\DA{}) algorithm of \citet{GS62college} proceeds in rounds. In each round, the algorithm consists of a \emph{proposal} phase, where each man who is currently unmatched proposes to his favorite woman from among those who haven't rejected him yet, followed by a \emph{rejection} phase where each woman tentatively accepts her favorite proposal and rejects the rest. The algorithm terminates when no further proposals can be made. 

\citet{GS62college} showed that given any profile $\>$ as input, the \DA{} algorithm always returns a stable matching as output; we denote this matching by $\DA(\>)$. They also observed that this matching is \emph{men-optimal}, i.e., it assigns each man his favorite stable partner among all stable matchings in $S_{\>}$. \citet{MW71stable} subsequently showed that this matching is also \emph{women-pessimal}.

\begin{restatable}[\protect\citealp{GS62college,MW71stable}]{proposition}{DAstableMenOptimalWomenPessimal}
Let $\>$ be a preference profile and let ${\mu \coloneqq \DA(\>)}$. Then, $\mu \in S_{\>}$. Furthermore, for any $\mu' \in S_{\>}$, $\mu \succeq_M \mu'$ and $\mu' \succeq_W \mu$.
\label{prop:DA_stable_MenOptimal_WomenPessimal}
\end{restatable}

\paragraph{Accomplice manipulation} Under this model of strategic behavior, a woman $w$, instead of misreporting herself, has a man $m$ provide a manipulated preference list, say $\>'_m$, in order to improve her match. Given a preference profile $\>$, we say that $w$ can \emph{manipulate through the accomplice} $m$ if $\mu'(w) \>_w \mu(w)$, where $\mu \coloneqq \DA(\>)$, $\>' \coloneqq \{\>_{-m}, \>'_m\}$, and $\mu' \coloneqq \DA(\>')$. We will often refer to $(m,w)$ as the \emph{manipulating pair} (not to be confused with a blocking pair).

Throughout this paper, any manipulation will be assumed to be \emph{optimal} unless stated otherwise. That is, there exists no other list $\succ''_m$ for the accomplice $m$ such that $\mu''(w) \>_w \mu'(w)$, where $\>'' \coloneqq \{\>_{-m}, \>''_m\}$, and $\mu'' \coloneqq \DA(\>'')$.

\paragraph{No-regret and with-regret manipulation}
We say that the accomplice $m$ incurs \textit{regret} if his match worsens upon misreporting, i.e., $\mu(m) \>_m \mu'(m)$. It is known that the \DA{} algorithm is strategyproof for the proposing side~\citep{DF81machiavelli}, which means that no man can improve his match by unilaterally misreporting his preferences. Therefore, for any man $m \in M$ and for any misreport $\>'_m$, we have that $\mu(m) \succeq_m \mu'(m)$. Thus, equivalently, we say that man $m$ incurs regret if $\mu(m) \neq \mu'(m)$.

We will consider two models of accomplice manipulation in this paper: \emph{no-regret manipulation} wherein only those misreports $\>'_m$ are allowed under which $\mu(m) = \mu'(m)$, and \emph{with-regret manipulation} where the accomplice is allowed (but not required) to incur regret. Thus, any no-regret strategy is also a with-regret strategy. Recall that the misreport in \Cref{eg:Self_vs_Accomplice_Intro} was a no-regret manipulation.

\paragraph{Stability relaxations} For any preference profile $\>$ and a fixed man $m \in M$, we say that a matching $\mu$ is \emph{$m$-stable}~\citep{BH19partners} with respect to $\>$ if any blocking pair (if one exists) in $\mu$ involves the man $m$. That is, for any pair $(m',w')$ that blocks $\mu$ under $\>$, we have $m' = m$. Clearly, a stable matching is also $m$-stable. Under accomplice manipulation, it can be shown that any matching $\mu'$ that is stable with respect to the manipulated profile (in particular, when $\mu' = \DA(\>')$) is $m$-stable with respect to the true profile~$\>$ (\Cref{prop:mStabilityExtension}). We note that \Cref{prop:mStabilityExtension} strengthens a result of \citet{BH19partners} who proved a similar statement only for a \DA{} matching. The proof of this result, along with all other omitted proofs, can be found in the appendix.

\begin{restatable}{proposition}{mStabilityExtension}
Let $\>$ denote the true preference profile. For any man $m$, let ${\>' \coloneqq \{\>_{-m},\>'_m\}}$, and let $\mu' \in S_{\>'}$ be any matching that is stable with respect to $\>'$. Then, $\mu'$ is $m$-stable with respect to $\>$.
\label{prop:mStabilityExtension}
\end{restatable}

\subsection{Structural Observations}
\label{subsec:Structural_Observations}

\paragraph{Push up/push down operations}
Note that given a profile $\>$, the preference list of any man $m$ can be written as ${ \>_m = (\>_m^L, \mu(m), \>_m^R) }$, where $\mu = \DA(\>)$ and $\>_m^L$ (respectively, $\>_m^R$) is the set of women that $m$ prefers to (respectively, finds less preferable than) $\mu(m)$. Interestingly, the \DA{} outcome does not change even if each man $m$ arbitrarily permutes the parts of his list above and below his \DA{}-partner $\mu(m)$. This result, due to \citet{H06cheating}, is recalled below.

\begin{restatable}[\protect\citealp{H06cheating}]{proposition}{PermutingFalsifiedLists}
Let $\>$ be a preference profile and let $\mu \coloneqq \DA(\>)$. For any man $m \in M$, let ${ \>'_m \coloneqq (\pi^L(\>^L_m),\mu(m),\pi^R(\>^R_m)) }$, where $\pi^L$ and $\pi^R$ are arbitrary permutations of $\>^L_m$ and $\>^R_m$, respectively. Let $\>' \coloneqq \{\>_{-m},\>'_m\}$, and let $\mu' \coloneqq \DA(\>')$. Then, $\mu' = \mu$.
\label{prop:Permuting_Falsified_Lists}
\end{restatable}

\Cref{prop:Permuting_Falsified_Lists} considerably simplifies the structure of accomplice manipulations that we need to consider. Indeed, we can assume that any manipulated list $\>'_m$ is such that the relative ordering of agents in the parts above and below $\mu'(m)$ is the same as under the true list $\>_m$, where $\mu' \coloneqq \DA(\>')$ and $\>' \coloneqq \{\>_{-m},\>'_m\}$ are the post-manipulation \DA{} outcome and preference profile, respectively.

This observation implies that, without loss of generality, any manipulated list $\>'_m$ can be obtained from the true list $\>_m$ by only \emph{push up} and \emph{push down} operations, wherein a set of women is pushed up above the true match $\mu(m)$, and another disjoint set is pushed below $\mu(m)$. Importantly, no permutation or shuffling operation is required as part of the manipulation. Formally, starting with the true list ${\>_m = (\>_m^L, \mu(m), \>_m^R)}$, we say that man $m$ performs a \emph{push up} operation for a set $X \subseteq W$ if the new list is ${\>^{X\uparrow}_m \coloneqq (\>_m^L \cup X, \mu(m), \>_m^R \setminus X)}$. Likewise, a \emph{push down} operation of a set $Y \subseteq W$ results in ${\>^{Y\downarrow}_m \coloneqq (\>_m^L \setminus Y, \mu(m), \>_m^R \cup Y)}$.

For manipulation via push down operations only, \citet{H06cheating} has shown that the resulting matching is weakly improving for \emph{all} men. Together with the fact that the \DA{} algorithm is strategyproof for the proposing side (i.e., the men)~\citep{DF81machiavelli}, we get that the \DA{} partner of the accomplice remains unchanged after a push down operation.

\begin{restatable}[\protect\citealp{DF81machiavelli,H06cheating}]{proposition}{PushDown}
Let $\>$ be the true preference profile and let $\mu \coloneqq \DA(\>)$. For any subset of women $X \subseteq W$ and any fixed accomplice $m \in M$, let ${ \>' \coloneqq \{\>_{-m},\>_m^{X \downarrow}\} }$ and $\mu' \coloneqq \DA(\>')$. Then, $\mu' \succeq_M \mu$ and $\mu'(m) = \mu(m)$.
\label{prop:PushDown}
\end{restatable}

The effect of push down operations for the proposed-to side is the exact opposite, as the resulting matching makes all women weakly worse off.

\begin{restatable}{lemma}{PushDownWorseForWomen}
Let $\>$ be the true preference profile and let $\mu \coloneqq \DA(\>)$. For any subset of women $X \subseteq W$, let ${ \>' \coloneqq \{\>_{-m},\>_m^{X \downarrow}\} }$ and $\mu' \coloneqq \DA(\>')$. Then, $\mu \succeq_W \mu'$.
\label{lem:PushDown_Worse_For_Women}
\end{restatable}

\Cref{lem:PushDown_Worse_For_Women} shows that in order to improve the partner of the woman~$w$, the use of \emph{push up} operations (by the accomplice) is necessary. However, it is not obvious upfront whether push up alone suffices; indeed, it is possible that the optimal strategy involves some combination of push up and push down operations. To this end, our theoretical results will show that, somewhat surprisingly, pushing up \emph{at most one} woman achieves the desired optimal manipulation (\Cref{thm:Inconspicuous_NoRegret,thm:Inconspicuous_WithRegret}). This strategy is known in the literature as \emph{inconspicuous manipulation}, which we define next.

\paragraph{Inconspicuous manipulation}
Given a profile $\>$ of true preferences and any fixed accomplice $m$, the manipulated list $\>'_m$ is said to be an \emph{inconspicuous manipulation} if the list $\>'_m$ can be derived from the true preference list $\>_m$ by promoting exactly one woman and making no other changes. The notion of inconspicuous manipulation has been previously studied in the context of self manipulation (where $w$ misreports herself), where it was shown that an optimal self manipulation is, without loss of generality, inconspicuous~\citep{VG17manipulating,DST18coalitional}.

\section{Theoretical Results}

\subsection{No-Regret Accomplice Manipulation} \label{sec:no-regret}

Let us start by observing that the \DA{} matching after an \emph{arbitrary} (optimal) no-regret accomplice manipulation may not be stable with respect to the true preferences.

\begin{example}
Consider the following preference profile where the \DA{} outcome is underlined.

\begin{table}[h]
    \centering
    \begin{tabularx}{0.9\linewidth}{XXXXXXXXXXXXXXX}
            $m_1\colon$ & $w_2^*$ & $\underline{w_1^\dagger}$ & $w_3$ & $w_4$ & $w_5$ && $w_1\colon$ & $\underline{m_1^\dagger}$ & $m_3$ & $m_2^*$ & $m_4$ & $m_5$\\
            $m_2\colon$ & $w_1^*$ & $\underline{w_2^\dagger}$ & $w_3$ & $w_4$ & $w_5$ && $w_2\colon$ & $\underline{m_2^\dagger}$ & $m_1^*$ & $m_3$ & $m_4$ & $m_5$\\
            $\boldsymbol{\textcolor{blue}{m_3}}\colon$ & $w_1$ & $\underline{w_3^{*,\dagger}}$ & $w_4$ & $w_2$ & $w_5$ && $w_3\colon$ & $\underline{m_3^{*,\dagger}}$ & $m_1$ & $m_2$ & $m_4$ & $m_5$\\
            $m_4\colon$ & $\underline{w_4}$ & $w_5^{*,\dagger}$ & $w_1$ & $w_2$ & $w_3$ && $\boldsymbol{\textcolor{blue}{w_4}}\colon$ & $m_5^{*,\dagger}$ & $m_3$ & $m_1$ & $m_2$ & $\underline{m_4}$\\
            $m_5\colon$ & $\underline{w_5}$ & $w_4^{*,\dagger}$ & $w_1$ & $w_2$ & $w_3$ && $w_5\colon$ & $m_4^{*,\dagger}$ & $m_1$ & $m_2$ & $m_3$ & $\underline{m_5}$
        \end{tabularx}
\end{table}

Suppose the manipulating pair is $(m_3, w_4)$. The \DA{} matches after the accomplice $m_3$ submits the manipulated list $\succ'_{m_3} \coloneqq w_4 \> w_3 \> w_1 \> w_2 \> w_5$ are marked by $*$. The manipulation results in $w_4$ to be matched with her top choice $m_5$ (i.e., $\>'_{m_3}$ is an optimal manipulation), an improvement over her true match $m_4$. Although $m_3$ does not incur regret, the manipulated matching admits a blocking pair $(m_3, w_1)$ with respect to the true preferences.
\label{eg:unstable-no-regret-manipulation}
\end{example}

Notice that if $m_3$ were to instead submit $\succ''_{m_3} \coloneqq w_4 \> w_1 \> w_3 \> w_2  \> w_5$ as his preference list in \Cref{eg:unstable-no-regret-manipulation}, then the resulting \DA{} matching (indicated by $\dagger$) would be stable with respect to the \emph{true} preferences while still allowing $w_4$ to match with $m_5$ (i.e., $\>''_{m_3}$ is also optimal). The manipulated list $\succ''_{m_3}$ is derived from the true list $\succ_{m_3}$ through a no-regret push up operation. Our first main result of this section (\Cref{thm:No-regret_StableLatticeContainment}) shows that this is not a coincidence: The set of all stable matchings with respect to a profile after a no-regret \emph{push up} operation is always contained within the stable set of the true preference profile.

\begin{restatable}[\textbf{No-regret push up is stability preserving}]{theorem}{StableLatticeContainment}
Let $\succ$ be a preference profile, and let $\mu \coloneqq \DA(\>)$. For any subset of women $X \subseteq W$ and any man $m$, let ${ \succ' \coloneqq \{ \succ_{-m}, \succ_m^{X\uparrow} \} }$, and $\mu' \coloneqq \DA(\>')$. If $m$ does not incur regret, then $S_{\succ'} \subseteq S_{\succ}$.
\label{thm:No-regret_StableLatticeContainment}
\end{restatable}

A primary consequence of \Cref{thm:No-regret_StableLatticeContainment} is that the \DA{} matching after a no-regret accomplice manipulation is weakly preferred over the true \DA{} outcome by all women, while the opposite is true for the men.

\begin{restatable}{corollary}{PushUp}
Let $\>$ be a preference profile and let ${\mu \coloneqq \DA(\>)}$. For any man $m$, let $\>' \coloneqq \{ \succ_{-m}, \succ_m^{X\uparrow} \}$ and $\mu' \coloneqq \DA(\>')$. If $m$ does not incur regret, then $\mu' \succeq_W \mu$ and $\mu \succeq_M \mu'$.
\label{cor:PushUp}
\end{restatable}
\begin{proof}
Since $m$ does not incur regret, it follows from \Cref{thm:No-regret_StableLatticeContainment} that $\mu' \in S_{\succ}$. Then, from \Cref{prop:DA_stable_MenOptimal_WomenPessimal}, we have that $\mu' \succeq_W \mu$ and $\mu \succeq_M \mu'$.
\end{proof}

As observed in \Cref{subsec:Structural_Observations}, any manipulation by the accomplice can be, without loss of generality, assumed to comprise only of push up and push down operations. We will now show that combining these operations is not necessary. That is, any manipulation that is achieved by a combination of push up and push down operations can be weakly improved by a push up operation alone (\Cref{lem:CombiningPushUpPushDown}). We note that this result does not require the no-regret assumption, and therefore applies to the with-regret setting as well.

\begin{restatable}{lemma}{CombiningPushUpPushDown}
Let $(m, w)$ be a manipulating pair and let $\>$ be a preference profile. For any subsets of women $X \subseteq W$ and $Y \subseteq W$, let $\>' \coloneqq \{ \succ_{-m}, \succ_m^{X\uparrow} \}$ denote the preference profile after pushing up the set $X$ and $\>'' \coloneqq \{ \succ_{-m}, \succ_m^{X\uparrow, Y\downarrow} \}$ denote the profile after pushing up $X$ and pushing down $Y$ in the true preference list $\>_m$ of man $m$. Let $\mu \coloneqq \DA(\>)$, $\mu' \coloneqq \DA(\>')$, and $\mu'' \coloneqq \DA(\>'')$. If $\mu''(w) \succ_w \mu(w)$, then $\mu'(w) \succeq_w \mu''(w)$.
\label{lem:CombiningPushUpPushDown}
\end{restatable}

Having narrowed down the strategy space to push up operations alone, we will now turn our attention to \emph{inconspicuous} manipulations (recall that such a manipulation involves promoting exactly one woman in the accomplice's true preference list to a higher position). We will show that any match for the manipulating woman $w$ that can be obtained by pushing up a \emph{set} of women can also be achieved by promoting \emph{exactly one} woman in that set (\Cref{lem:PushUpOneWoman}). In other words, any no-regret push up operation is, without loss of generality, inconspicuous. We note that although \Cref{lem:PushUpOneWoman} assumes no regret for the accomplice, the corresponding implication actually holds even in the with-regret setting (see \Cref{lem:PushUpOneWomanWithRegret}).

\begin{restatable}{lemma}{PushUpOneWoman}
Let $(m, w)$ be a manipulating pair, and let ${X \subseteq W}$ be an arbitrary set of women that $m$ can push up without incurring regret. Then, the match for $w$ that is obtained by pushing up all women in $X$ can also be obtained by pushing up exactly one woman in $X$.
\label{lem:PushUpOneWoman}
\end{restatable}

We will now use the foregoing observations to prove our main result (\Cref{thm:Inconspicuous_NoRegret}).

\begin{restatable}{theorem}{InconspicuousNoRegret}
Any optimal no-regret accomplice manipulation is, without loss of generality, inconspicuous.
\label{thm:Inconspicuous_NoRegret}
\end{restatable}

\begin{proof}
From \Cref{prop:Permuting_Falsified_Lists} (and subsequent remarks), we know that any accomplice manipulation can be simulated via push up and push down operations. \Cref{lem:CombiningPushUpPushDown} shows that any beneficial manipulation that is achieved by some combination of pushing up a set $X \subseteq W$ and pushing down $Y \subseteq W$ can be weakly improved by only pushing up $X \subseteq W$. Finally, from \Cref{lem:PushUpOneWoman}, we know that any match for the manipulating woman $w$ that is achieved by pushing up $X \subseteq W$ is also achieved by pushing up exactly one woman in $X$, thus establishing the desired inconspicuousness property.
\end{proof}

\Cref{thm:Inconspicuous_NoRegret} has some interesting computational and structural implications. First, the inconspicuousness property implies a straightforward polynomial-time algorithm for computing an optimal no-regret accomplice manipulation (\Cref{cor:Optimal_No-regret_PolynomialTime}). Second, the \DA{} matching resulting from an inconspicuous no-regret manipulation is stable with respect to the true preferences (\Cref{cor:No-regret_Inconspicuous_Stable}). Together, these results reconcile the seemingly conflicting interests of the manipulator (who wants to compute optimal manipulation efficiently) and the central planner (who wants the resulting matching to be stable with respect to the true preferences).

\begin{restatable}{corollary}{PolynomialTimeNoRegret}
An optimal no-regret accomplice manipulation strategy can be computed in $\O(n^3)$ time.
\label{cor:Optimal_No-regret_PolynomialTime}
\end{restatable}

\begin{restatable}{corollary}{InconspicuousStable}
The \DA{} outcome from an inconspicuous no-regret accomplice manipulation is stable with respect to the true preferences.
\label{cor:No-regret_Inconspicuous_Stable}
\end{restatable}

\begin{remark}
Recall from \Cref{eg:unstable-no-regret-manipulation} that an arbitrary optimal no-regret strategy may not be stability-preserving. Nevertheless, any optimal no-regret strategy admits an equivalent inconspicuous strategy (\Cref{thm:Inconspicuous_NoRegret}) which indeed preserves stability~(\Cref{cor:No-regret_Inconspicuous_Stable}).
\end{remark}

\subsection{With-Regret Accomplice Manipulation}\label{sec:with-regret}

No-regret manipulations come at no cost for the accomplice and thus are a viable strategic behavior (as shown in \Cref{fig:AvTP_SvTP_AvS_SvA}). Yet, a more permissive strategy space may allow for the accomplice to incur some regret. Such \emph{with-regret} manipulations may be justifiable in practice: An accomplice's idiosyncratic preference may be tolerant to a small loss in exchange of gain for the partnering woman, or a woman may persuade a man to withstand some regret by providing side-payments.

We will start by illustrating that a with-regret accomplice manipulation can be strictly more beneficial compared to its no-regret and self manipulation counterparts.

\begin{example}[\textbf{With-regret vs. no-regret}]
Consider the following preference profile where the \DA{} outcome is underlined.

\begin{table}[h]
    \centering
    \begin{tabularx}{\linewidth}{XXXXXXXXXXXXXXX}
            $\boldsymbol{\textcolor{blue}{m_1}}\colon$ & \underline{$w_4^*$} & $w_1^\dagger$ & $w_2$ & $w_5$ & $w_3$ && $\boldsymbol{\textcolor{blue}{w_1}}\colon$ & $m_1^\dagger$ & $m_2^*$ & \underline{$m_3$} & $m_4$ & $m_5$\\
            $m_2\colon$ & \underline{$w_2$} & $w_4$ & $w_1^*$ & $w_5^\dagger$ & $w_3$ && $w_2\colon$ & $m_3^{*,\dagger}$ & $m_5$ & $m_1$ & \underline{$m_2$} & $m_4$\\
            $m_3\colon$ & \underline{$w_1$} & $w_2^{*,\dagger}$ & $w_4$ & $w_3$ & $w_5$ && $w_3\colon$ & $m_2$ & \underline{$m_5^*$} & $m_1$ & $m_4^\dagger$ & $m_3$\\
            $m_4\colon$ & $w_1$ & $w_3^\dagger$ & \underline{$w_5^*$} & $w_2$ & $w_4$ && $w_4\colon$ & $m_4$ & $m_3$ & \underline{$m_1^*$} & $m_5^\dagger$ & $m_2$\\
            $m_5\colon$ & $w_1$ & $w_4^\dagger$ & \underline{$w_3^*$} & $w_5$ & $w_2$ && $w_5\colon$ & \underline{$m_4^*$} & $m_2^\dagger$ & $m_5$ & $m_1$ & $m_3$
        \end{tabularx}
\end{table}

Suppose the manipulating pair is $(m_1, w_1)$. The \DA{} matching after $m_1$ submits the optimal no-regret\footnote{To see why $\>_{m_1}'$ is an \emph{optimal} no-regret manipulation, note that the woman-optimal stable matching (with respect to $\>$) matches $w_1$ with $m_2$. From \Cref{thm:Inconspicuous_NoRegret,cor:No-regret_Inconspicuous_Stable}, an optimal no-regret manipulation is, without loss of generality, stability preserving, and from \Cref{prop:DA_stable_MenOptimal_WomenPessimal}, $m_2$ is the best stable partner for $w_2$.} manipulated list $\>_{m_1}' \coloneqq w_2 \> w_4 \> w_1 \> w_5 \> w_3$ and the optimal with-regret manipulated list $\>_{m_1}'' \coloneqq w_1 \> w_4 \> w_2 \> w_5 \> w_3$ are marked by $*$ and $\dagger$, respectively. Both manipulation strategies improve $w_1$'s matching compared to truthful reporting, but $w_1$ strictly prefers the with-regret outcome.
\label{eg:WithRegret}
\end{example}

\Cref{eg:WithRegret} highlights two key differences between optimal no-regret and with-regret manipulations. First, the matching after the \textit{inconspicuous} with-regret manipulation (marked by $\dagger$) admits a blocking pair $(m_1, w_4)$ with respect to the true profile $\>$. This is in contrast to the no-regret case which is stability preserving
(\Cref{thm:No-regret_StableLatticeContainment}).
Second, in contrast to \Cref{cor:PushUp}, an optimal with-regret manipulation is not guaranteed to weakly improve or worsen the matching for \emph{all} agents on one side; indeed the women $w_3$ and $w_5$ are strictly worse off while $w_1$ is strictly better off. Similarly, the man $m_1$ is strictly worse off while $m_4$ and $m_5$ strictly improve.

The primary distinction between no-regret and with-regret manipulation lies in the push up operations. If pushing up a set of women does not incur regret for the accomplice then pushing up any subset thereof does not either. In contrast, if by pushing up a set of women the accomplice incurs regret, then there exists exactly one woman in that set who causes the same level of regret when pushed up individually.  As previously mentioned, with-regret push up operations do not uniformly affect all men and all women (in contrast to \Cref{cor:PushUp}). Moreover, the set of attained matchings after a with-regret manipulation are no longer stable with respect
to true preferences (in contrast to \Cref{thm:No-regret_StableLatticeContainment}), which makes the analysis challenging.

Despite these structural differences, we are able to prove an analogue of \Cref{lem:PushUpOneWoman} for with-regret push up operations (\Cref{lem:PushUpOneWomanWithRegret}). Our proof of this result relies on the fact that all proposals that occur when the accomplice pushes up a set of women are contained in the union of sets of proposals that occur when pushing up individual women in that set. This is relatively easy to prove for the no-regret case, since the \DA{}
matchings after these push up operations are all stable with respect to true preferences (\Cref{thm:No-regret_StableLatticeContainment}). Although we cannot rely on the same stability result for the with-regret case, we circumvent the issue by reasoning about the sets of proposals in greater detail. The full proof of \Cref{lem:PushUpOneWomanWithRegret}, along with an extensive discussion, is relegated to the appendix.

\begin{restatable}{lemma}{PushUpOneWomanWithRegret}
Let $(m, w)$ be a manipulating pair, and let ${X \subseteq W}$ be an arbitrary set of women that $m$ can push up (while incurring regret). Then, the match for $w$ that is obtained by pushing up all women in $X$ can also be obtained by pushing up exactly one woman in $X$.
\label{lem:PushUpOneWomanWithRegret}
\end{restatable}

Subsequently, an optimal with-regret manipulation is, without loss of generality, inconspicuous.
The proof is similar to that of the no-regret case (\Cref{thm:Inconspicuous_NoRegret}) with the only difference being the use of \Cref{lem:PushUpOneWomanWithRegret} in place of \Cref{lem:PushUpOneWoman}. 

\begin{restatable}{theorem}{InconspicuousWithRegret}
Any optimal with-regret accomplice manipulation is, without loss of generality, inconspicuous.
\label{thm:Inconspicuous_WithRegret}
\end{restatable}

\Cref{thm:Inconspicuous_WithRegret} immediately implies a polynomial-time algorithm for computing an optimal with-regret accomplice manipulation. Moreover, the \DA{} outcome from any inconspicuous with-regret accomplice manipulation is $m$-stable with respect to the true preferences (\Cref{prop:mStabilityExtension}).

\begin{restatable}{corollary}{PolynomialTimeWithRegret}
An optimal with-regret accomplice manipulation strategy can be computed in $\O(n^3)$ time.
\label{cor:Optimal_With-regret_PolynomialTime}
\end{restatable}

\section{Experimental Results}
\label{sec:Experiments}

In addition to the experiments described in Section~\ref{sec:Introduction}, we performed a series of simulations to analyze the performance of accomplice manipulation. As for the previous experimental setup, we generated 1000 profiles uniformly at random for each value of $n \in \{3,\dots, 40\}$ (where $n$ is the number of men/women) and allowed any man to be chosen as an accomplice for each experiment unless stated otherwise.

\paragraph{Comparing the Quality of Partners}

We first compare the quality of partners that a fixed strategic woman $w$ is matched with through no-regret accomplice and self manipulation. \Cref{fig:RankDiff_Accomplice_vs_Self} illustrates the distributions of improvement (in terms of rank difference) out of only those instances where $w$ is strictly better off through the two strategies individually. In other words, the self (respectively, accomplice) manipulation boxplots only reflect the data for when self (respectively, accomplice) manipulation is successful. It is evident that, in expectation, $w$ is matched with better partners through no-regret accomplice manipulation.

\begin{figure}[h]
    \centering
    \includegraphics[width=0.72\linewidth]{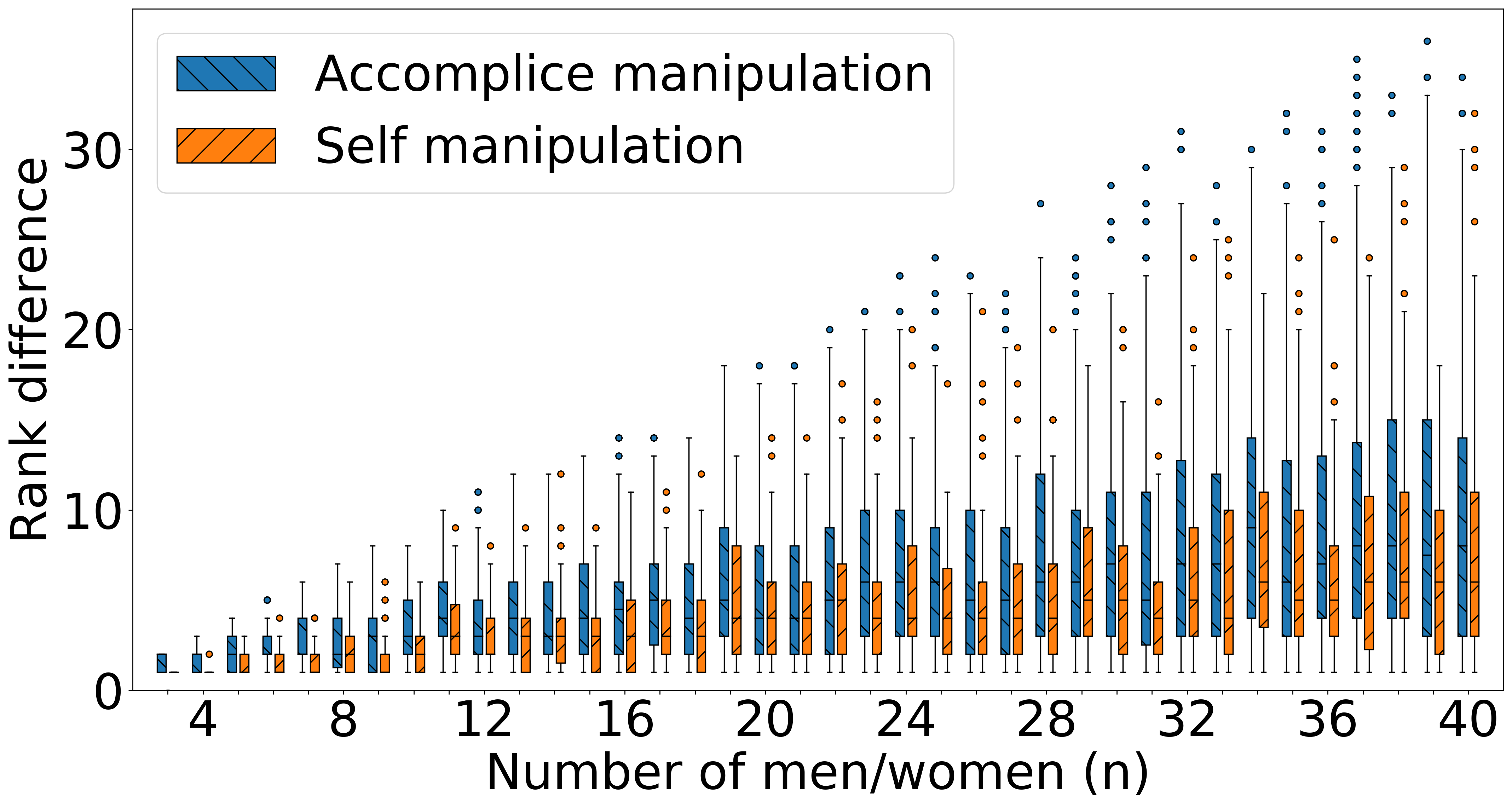}
    \caption{Comparing no-regret accomplice and self manipulation in terms of the \emph{improvement in the rank of the matched partner} of $w$. The solid bars, whiskers, and dots denote the interquartile range, range excluding outliers, and outliers, respectively.}
\label{fig:RankDiff_Accomplice_vs_Self}
\end{figure}

\paragraph{The Fraction of Women Who Improve}

We additionally compare the fraction of women who are able to improve through no-regret accomplice and self manipulation individually. \citet{TS01gale} reported that 5.06\% of women were able to improve using self manipulation when $n = 8$. In our experiment, this value is similarly 4.18\%. However, 9.99\% of women are able to improve through no-regret accomplice manipulation. As illustrated in Figure~\ref{fig:FractionWomen_AvS}, the fraction of women who benefit from no-regret accomplice manipulation is consistently more than double that of self manipulation.

\begin{figure}[h]
    \centering
    \includegraphics[width=0.72\linewidth]{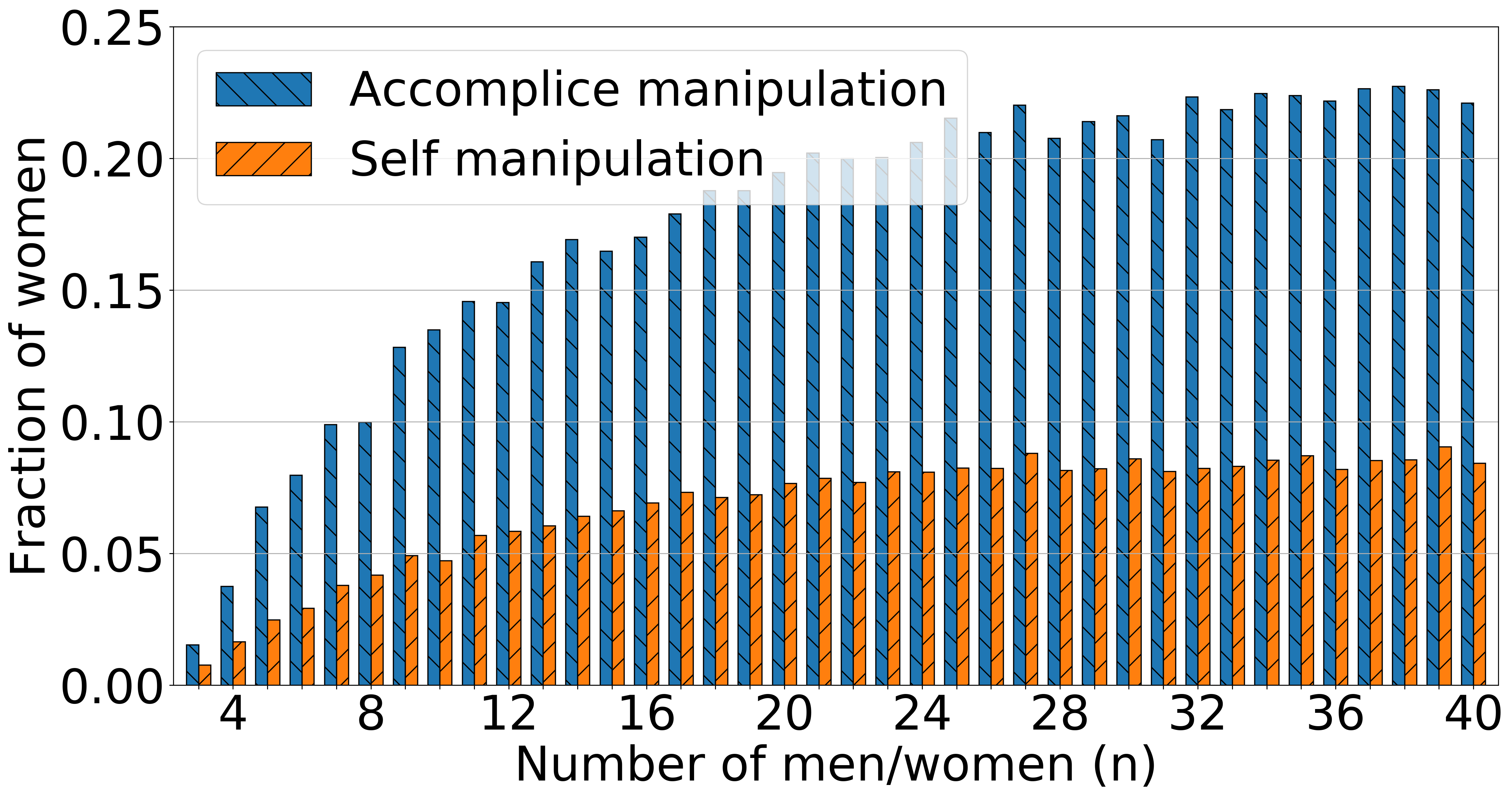}
    \caption{Comparing no-regret accomplice and self manipulation in terms of the \emph{fraction of women who benefit}.}
\label{fig:FractionWomen_AvS}
\end{figure}

\paragraph{How Much Flexibility is Really Needed in Choosing the Accomplice?}

So far, our experiments have taken the optimistic approach of allowing the strategic woman to pick \emph{any} man of her choice as the accomplice. To examine the exact extent of flexibility that this assumption requires, we conduct an experiment where the accomplice is chosen from a \emph{fixed} pool of $p$ men for some $p \leq n$ (for example, when $p = 5$, we pick the best accomplice from a fixed set of five men). For $n = 40$, we ran the no-regret accomplice, with-regret accomplice, and self manipulation strategies on 1000 profiles for $p \in \{1,\dots,40\}$. The results are presented in \Cref{fig:VaryAccomplicePools}.

\begin{figure}[h]
    \centering
    \includegraphics[width=0.72\linewidth]{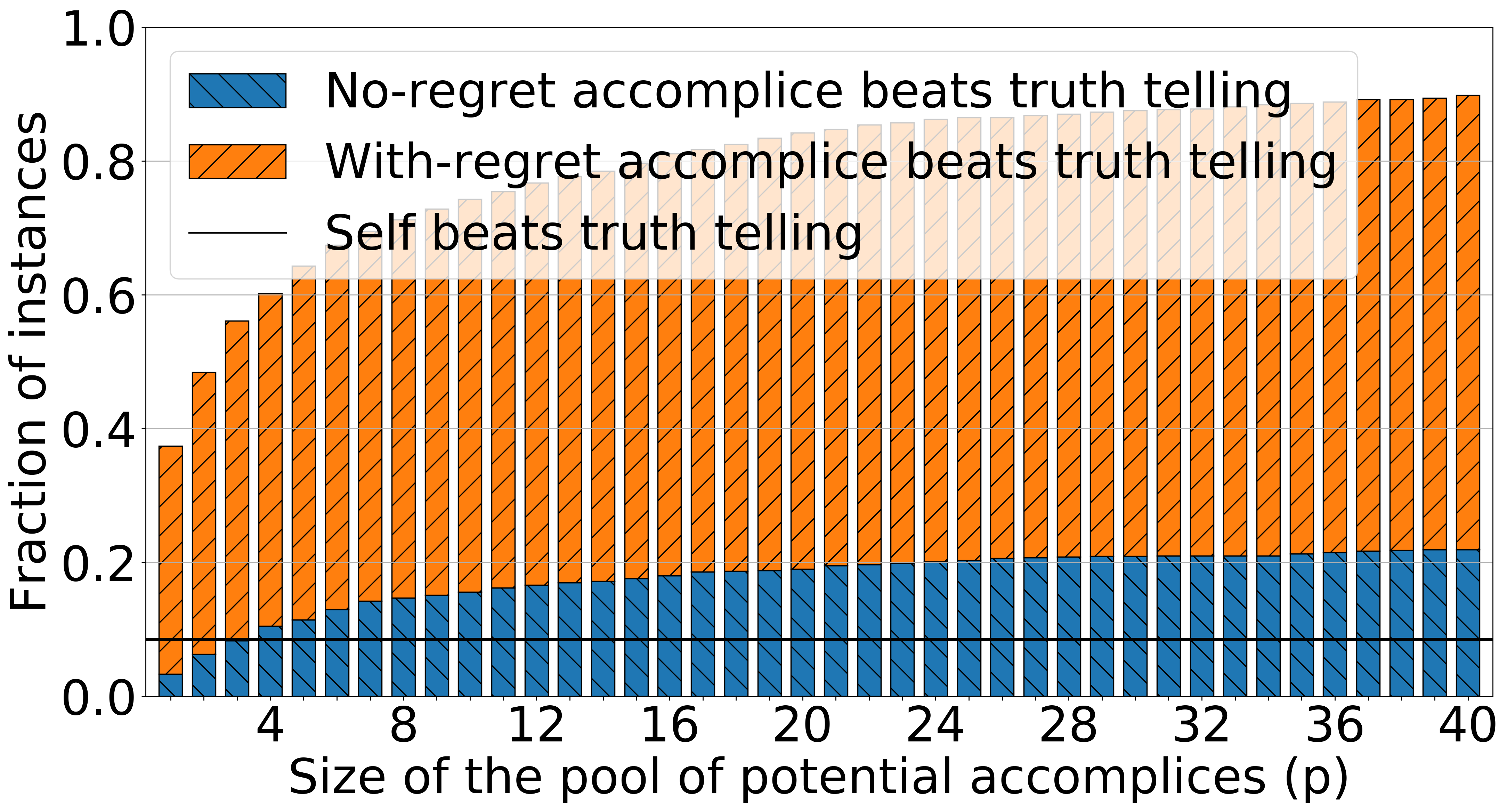}
    \caption{Comparing no-regret accomplice, with-regret accomplice (with variable-sized pools of potential accomplices for both), and self manipulation against truthful reporting when $n = 40$.}
\label{fig:VaryAccomplicePools}
\end{figure}

One would expect with-regret accomplice manipulation to outperform self manipulation when there is variability in accomplices. Indeed, the strategic woman could simply request her top choice to place her at the top of his list. However, we observe that, in expectation, with-regret accomplice manipulation outperforms self manipulation for \emph{every} pool size $p$ (thus, a with-regret strategy outperforms self manipulation even when a single accomplice is randomly chosen in advance). Furthermore, the comparatively limited no-regret accomplice manipulation is also, on average, better than self manipulation when there are at least four men to choose from. These observations suggest that the superior performance of accomplice manipulation can be achieved even with a modest amount of flexibility in the choice of the accomplice.

\section{Concluding Remarks}

We showed that accomplice manipulation is a viable strategic behavior that only requires inconspicuous misreporting of preferences and is frequently more beneficial than the classical self-manipulation strategy. A natural avenue for future research is to investigate a setting with \emph{multiple} accomplices working together to manipulate the outcome for the strategic woman. Additionally, one can think of a broader strategy space in which both the accomplice and the manipulating woman are able to misreport their preference lists simultaneously. Analyzing the benefits of such \emph{coalitional} manipulation strategies---with or without regret---on one or both sides, and studying their structural and algorithmic properties are intriguing directions for future work.

\section*{Acknowledgments}
HH acknowledges support from NSF grant \#1850076. RV acknowledges support from ONR\#N00014-171-2621 while he was affiliated with Rensselaer Polytechnic Institute, and is currently supported by project no. RTI4001 of the Department of Atomic Energy, Government of India. Part of this work was done while RV was supported by the Prof. R Narasimhan postdoctoral award.

\bibliographystyle{named}
\bibliography{ms}

\clearpage
\begin{center}
\Large{Appendix}
\end{center}

\section{Omitted Material from Sections~\ref{sec:Preliminaries_Stable_Matching} \& \ref{subsec:Structural_Observations}}

\subsection{Additional Preliminaries}
\label{subsec:Additional_Preliminaries}

\paragraph{Lattice of stable matchings} Given any preference profile $\>$ and any pair of stable matchings $\mu, \mu' \in S_{\>}$, let us define their \emph{meet} $\mu_{\wedge} \coloneqq \mu \wedge \mu'$ as follows: For every man $m \in M$,
\[
\mu_{\wedge}(m) = 		
    \begin{cases}
        \mu'(m) & \text{if } \mu(m) \>_m \mu'(m) \\
        \mu(m) & \text{otherwise,}
    \end{cases}
\]
  and for every women $w \in W$,
\[
\mu_{\wedge}(w) = 		
    \begin{cases}
        \mu(w) & \text{if } \mu(w) \>_w \mu'(w) \\
        \mu'(w) & \text{otherwise.}
    \end{cases}
\]
  Similarly, the \emph{join} $\mu_{\vee} \coloneqq \mu \vee \mu'$ is defined as follows: For every man $m \in M$,
\[
\mu_{\vee}(m) = 		
    \begin{cases}
        \mu(m) & \text{if } \mu(m) \>_m \mu'(m) \\
        \mu'(m) & \text{otherwise,}
    \end{cases}
\]
 and for every women $w \in W$,
\[
\mu_{\vee}(w) = 
    \begin{cases}
        \mu'(w) & \text{if } \mu(w) \>_w \mu'(w) \\
        \mu(w) & \text{otherwise.}
    \end{cases}
\]
  
A well-known result, attributed to John Conway~\citep{K97stable}, establishes that the set of stable matchings is closed under meet and join operations.

\begin{restatable}[\protect\citealp{K97stable}]{proposition}{Lattice}
Let $\>$ be a preference profile and let $\mu,\mu' \in S_{\>}$. Then, $\mu_{\wedge},\mu_{\vee} \in S_{\>}$.
\label{prop:Lattice}
\end{restatable}

\subsection{Self Manipulation vs. No-Regret Accomplice Manipulation}

The self manipulation and no-regret accomplice manipulation strategy spaces are not contained in one another. \Cref{eg:Self_vs_Accomplice_Intro} shows an instance where no-regret accomplice manipulation is better than self manipulation. In \Cref{eg:Self_beats_Accomplice}, we provide an instance where self manipulation is better than no-regret accomplice manipulation. 

\begin{example}
Consider the following preference profile where the \DA{} outcome is underlined.

\begin{table}[h]
\centering 
    \begin{tabular}{p{0.019\textwidth}>{\centering}p{0.019\textwidth}>{\centering}p{0.019\textwidth}>{\centering}p{0.019\textwidth}>{\centering}p{0.019\textwidth}>{\centering}p{0.032\textwidth}>{\centering}p{0.019\textwidth}>{\centering}p{0.019\textwidth}>{\centering}p{0.019\textwidth}>{\centering}p{0.019\textwidth}>{\centering\arraybackslash}p{0.019\textwidth}}
         	$m_1$: & \underline{$w_2$} & $w_3$ &
         	$w_1^*$ & $w_4$ && $\boldsymbol{\textcolor{blue}{w_1}}$: & $m_1^*$ & $m_2$ & \underline{$m_3$} & $m_4$\\
            $m_2$: & \underline{$w_3$} & $w_2^{*}$ & $w_4$ & $w_1$ && $w_2$: & $m_2^{*}$ & $m_3$ & $m_4$ & \underline{$m_1$}\\
            $m_3$: & \underline{$w_1$} & $w_3^*$ & $w_4$ & $w_2$ && $w_3$: & $m_3^*$ & $m_1$ & $m_4$ & \underline{$m_2$}\\
            $m_4$: & $w_1$ & \underline{$w_4^*$} & $w_2$ & $w_3$ && $w_4$: & $m_3$ & $m_1$ & \underline{$m_4^*$} & $m_2$
    \end{tabular}
\end{table}

Suppose $w_1$ seeks to improve her match via manipulation. 
The optimal self manipulation strategy for $w_1$ is $\>'_{w_1} = m_4 \> m_3 \> m_1 \> m_2$, which allows her to match with her top choice $m_1$ (the new \DA{} matching is marked by $*$). Regardless of the choice of accomplice, the optimal no-regret accomplice manipulation strategy, on the other hand, is truth-telling.
\label{eg:Self_beats_Accomplice}
\end{example}

\subsection{Proof of Proposition~\ref{prop:mStabilityExtension}}

\mStabilityExtension*
\begin{proof}
Suppose, for contradiction, that $\mu'$ is not $m$-stable with respect to $\>$. Then, there must exist a man-woman pair $(m', w')$ that blocks $\mu'$ with respect to $\>$ such that $m' \neq m$, i.e., $w' \>_{m'} \mu'(m')$ and $m' \>_{w'} \mu'(w')$. Since $m$ is the only agent whose preferences differ between $\>$ and $\>'$, we have that $\>_{m'} \, = \, \>'_{m'}$ and $\>_{w'} \, = \, \>'_{w'}$. Thus, $w' \>'_{m'} \mu'(m')$ and $m' \>'_{w'} \mu'(w')$, implying that the pair $(m',w')$ blocks $\mu'$ with respect to $\>'$, which contradicts the assumption that $\mu' \in S_{\>'}$. Thus, $\mu'$ must be $m$-stable with respect to the true profile $\>$.
\end{proof}

\subsection{Proof of Lemma~\ref{lem:PushDown_Worse_For_Women}}

\PushDownWorseForWomen*
\begin{proof}
Suppose, for contradiction, that there exists a woman $w'$ such that $m' \succ_{w'} \mu(w')$, where $m' \coloneqq \mu'(w')$. We infer that $w' \nsucc_{m'} \mu(m')$, otherwise the stability of $\mu$ with respect to $\succ$ is compromised. Since $m' \neq \mu(w')$, we have that $w' \neq \mu(m')$, and therefore $\mu(m') \>_{m'} w'$. However, \Cref{prop:PushDown} guarantees $\mu' \succeq_M \mu$, thus posing a contradiction.
\end{proof}

\section{Omitted Material from Section~\ref{sec:no-regret}}

\subsection{Proof of Theorem~\ref{thm:No-regret_StableLatticeContainment}}

\StableLatticeContainment*
\begin{proof}
Suppose, for contradiction, that there exists a matching $\phi \in S_{\>'} \setminus S_{\>}$. Then, there must be a pair $(m', w')$ that blocks $\phi$ with respect to $\>$. It follows from \Cref{prop:mStabilityExtension} that $m' = m$. Thus, $w' \succ_{m} \phi(m)$ and $m \succ_{w'} \phi(w')$.

From \Cref{prop:DA_stable_MenOptimal_WomenPessimal}, we have that $\mu'(m) \succeq'_{m} \phi(m)$. Since $m$ does not incur regret, we have $\mu(m) = \mu'(m)$, and thus, $\mu(m) \succeq'_{m} \phi(m)$. All women below $\mu(m)$ in $\succ'_m$ are also below $\mu(m)$ in $\succ_m$ by the push up assumption. Since $\mu(m) \succeq'_m \phi(m)$, this implies that all women below $\phi(m)$ in $\succ'_m$ are also below $\phi(m)$ in $\succ_m$. Thus, if $\phi(m) \>'_m w'$, then $\phi(m) \>_m w'$, which contradicts the blocking pair condition above. Therefore, we must have that $w' \>'_m \phi(m)$ (note that $w' \neq \phi(m)$ by the blocking pair condition). Furthermore, since $\succ'_{w'} \ = \ \succ_{w'}$, the blocking pair condition also implies that $m \succ'_{w'} \phi(w')$. Combined with the condition ${ w' \succ'_m \phi(m) }$, this contradicts the assumption that $\phi \in S_{\>'}$. Thus, $S_{\>'} \subseteq S_{\>}$.
\end{proof}

\subsection{Proof of Lemma~\ref{lem:CombiningPushUpPushDown}}

\CombiningPushUpPushDown*
\begin{proof}
Our proof will use case analysis based on whether or not $\mu'(m) = \mu(m)$.

\textbf{Case I} (when $\mu'(m) = \mu(m)$): The list $\>''_m$ can be considered as being derived from $\>'_m$ via a push down operation on the set $Y$ (see \Cref{fig:combining-push-up-push-down}). From \Cref{lem:PushDown_Worse_For_Women}, we know that a push down operation is weakly worse for all women; thus, in particular, we get $\mu'(w) \succeq_w \mu''(w)$, as desired. Note that the relative ordering of $X$ and $Y$ above $\mu'(m)$ in the list $\>'_m$ is not important in light of \Cref{prop:Permuting_Falsified_Lists}.

\begin{figure}[h]
\centering
\begin{tikzcd}[cramped, column sep=tiny]
    \succ_m: & \cdots & Y & \cdots & \mu(m) & \cdots & X & \cdots\\
    \succ'_m: & \cdots & X & \cdots & Y & \cdots & \mu(m) & \cdots\\
    \succ''_m: & \cdots & X & \cdots & \mu(m) & \cdots & Y & \cdots
\end{tikzcd}
\caption{An illustration of man $m$'s preference lists under the profiles $\succ$, $\succ'$, and $\succ''$ in the proof of \Cref{lem:CombiningPushUpPushDown}.}
\label{fig:combining-push-up-push-down}
\end{figure}
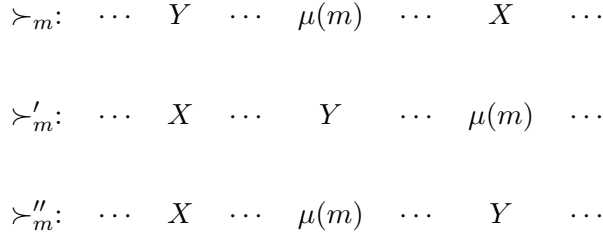

\textbf{Case II} (when $\mu'(m) \neq \mu(m)$): Suppose $\mu'(m) \in X$. Then, the list $\>''_m$ can be considered as being derived from $\>'_m$ via a permutation of the women below $\mu'(m)$ (see \Cref{fig:combining-push-up-push-down}). By \Cref{prop:Permuting_Falsified_Lists}, this implies that $\mu' = \mu''$, and in particular $\mu'(w) = \mu''(w)$, as desired.

Therefore, for the remainder of the proof, let us assume that $\mu'(m) \notin X$. Since $\>'_m$ is derived from $\>_m$ via a push up operation on the set $X$, and since $\mu'(m) \neq \mu(m)$ by assumption, we have that $\mu(m) \>_m \mu'(m)$. By \Cref{prop:Permuting_Falsified_Lists}, we can assume, without loss of generality, that $\mu'(m)$ is positioned immediately below $\mu(m)$ in the list $\>_m$. By construction, the same property also holds for the lists $\>'_m$ and $\>''_m$. Thus, $\>''_m$ can be considered as being obtained from $\>'_m$ via a push down operation on the set $Y$ (note that this operation is defined with respect to $\mu'(m)$). By \Cref{lem:PushDown_Worse_For_Women}, we have $\mu'(w) \succeq_w \mu''(w)$, as desired.
\end{proof}

\subsection{Proof of Lemma~\ref{lem:PushUpOneWoman}}

\PushUpOneWoman*

\paragraph{Preliminaries for the proof of \Cref{lem:PushUpOneWoman}:} Let $\>$ be a true preference profile. Given an accomplice $m$ and a set of women $X = \{w_a, w_b, w_c, \dots \}$ such that $\mu(m) \>_m w_x$ for all $w_x \in X$, we define $\>^a_m$ as the list derived from $\>_m$ where $m$ pushes up $w_a$, $\>^b_m$ as the list derived from $\>_m$ where $m$ pushes up $w_b$, etc., and $\>^X_m$ as the list where $m$ pushes up all women in $X$; the corresponding profiles are $\>^a, \>^b, \>^c, \dots, \>^X$. Additionally, let $\mu^a \coloneqq \DA(\>^a), \mu^b \coloneqq  \DA(\>^b), \mu^c \coloneqq  \DA(\>^c), \dots, \mu^X \coloneqq \DA(\>^X)$. Note that the placement of women being pushed up above $\mu(m)$ in $\>^a_m, \>^b_m, \>^c_m, \dots, \>^X_m$ does not affect $\mu^a, \mu^b, \mu^c, \dots, \mu^X$ by \Cref{prop:Permuting_Falsified_Lists}.

\begin{example}
Suppose $X = \{w_a, w_b\}$. \Cref{fig:inconspicuous-profiles} illustrates the possible configurations of $m$'s preference lists under the profiles $\>$, $\>^a$, $\>^b$, and $\>^X$.

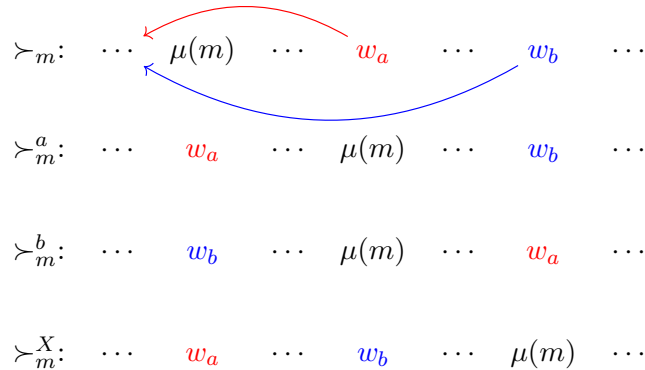
\begin{figure}[!htb]
\centering
\begin{minipage}[t]{.475\textwidth}
    \centering
    \begin{tikzcd}[cramped, column sep=tiny]
        \>_m: & |[alias=dots]| \cdots & \mu(m) & \cdots & \textcolor{red}{w_a} \arrow[to=dots, bend right=30, red] & \cdots & \textcolor{blue}{w_b} \arrow[to=dots, bend left=30, blue] & \cdots\\
        \>_m^a: & \cdots & \textcolor{red}{w_a} & \cdots & \mu(m) & \cdots & \textcolor{blue}{w_b} & \cdots\\
        \>_m^b: & \cdots & \textcolor{blue}{w_b} & \cdots & \mu(m) & \cdots & \textcolor{red}{w_a} & \cdots\\
        \>_m^X: & \cdots & \textcolor{red}{w_a} & \cdots & \textcolor{blue}{w_b} & \cdots & \mu(m) & \cdots
    \end{tikzcd}
\end{minipage}
\caption{An illustration of man $m$'s preference lists under the profiles $\>$, $\>^a$, $\>^b$, and $\>^X$ when $X = \{ w_a, w_b \}$.}
\label{fig:inconspicuous-profiles}
\end{figure}
\end{example}

It can be shown that $m$ does not incur regret under any of the profiles $\>^a, \>^b,$ and so on (\Cref{lem:SingleAgentPushUpNoRegret}).

\begin{restatable}{lemma}{SingleAgentPushUpNoRegret}
Let $X = \{w_a,w_b,w_c,\dots\}$ be an arbitrary finite set of women that the accomplice $m$ can push up without incurring regret. Then, for every $w_x \in X$, $m$ does not incur regret under the matching $\mu^x \coloneqq \DA(\>^x)$, where $\>^x \coloneqq \{\>_{-m},\>_m^{w_x\uparrow}\}$.
\label{lem:SingleAgentPushUpNoRegret}
\end{restatable}

\begin{proof}
We will prove the lemma for the fixed profile $\>^a$ (an identical argument works for other profiles).

Suppose, for contradiction, that $m$ incurs regret in the profile $\>^a$. That is, $\mu(m) \>_m \mu^a(m)$ where $\mu^a \coloneqq \DA(\>^a)$. From \Cref{lem:accomplice-match-after-regret-manipulation} (see \Cref{subsec:With_Regret_Inconspicuous_Proof}), we get that $\mu^a(m) = w_a$. Further, using \Cref{prop:Permuting_Falsified_Lists}, we can assume, without loss of generality, that $w_a$ is positioned immediately below $\mu(m)$ in the true list $\>_m$, and immediately above it in the manipulated lists $\>_m^{a}$, as well as $\>_m^{X}$. This implies that the transition from $\>_m^{a}$ to $\>_m^{X}$ is a \emph{with-regret} push up operation involving the promotion of $X \setminus \{w_a\}$ (since, according to the list $\>_m^{a}$, the new partner $\mu(m)$ is strictly worse than $w_a$). Again, from \Cref{lem:accomplice-match-after-regret-manipulation}, it follows that $\mu^X(m) \in X \setminus \{w_a\}$. By the no-regret assumption for the set $X$, we know that $\mu^X(m) = \mu(m)$. This, however, is a contradiction since all women in $X \setminus \{w_a\}$ are placed below $\mu(m)$ in the list $\>_m$, and hence must be different from it.
\end{proof}

Given any profile $\>$, let $P_{\>}$ denote the set of all proposals that occur in the execution of the \DA{} algorithm on $\>$. Formally, for any man $m_i \in M$ and woman $w_j \in W$, the ordered pair $(m_i,w_j)$ belongs to the set $P_{\>}$ if $m_i$ proposes to $w_j$ during the execution of \DA{} algorithm on the profile $\>$.

\begin{restatable}{lemma}{PushUpProposalsContainment}
Let $X = \{w_a,w_b,w_c,\dots\}$ be an arbitrary finite set of women that the accomplice $m$ can push up without incurring regret. Then, any proposal that occurs under $\>^X$ also occurs under at least one of the profiles $\>^{a}, \>^{b}, \>^{c}, \dots$.
\label{lem:PushUpPropsalsContainment}
\end{restatable}

\begin{proof}
Suppose, for contradiction, that $(m_1, w_1)$ is the \emph{first} proposal to occur during the \DA{} execution on $\>^X$ such that it is not an element of $P_{\>^a} \cup P_{\>^b} \cup P_{\>^c} \cup \dots$. Note that the proposals made by the accomplice $m$ in $P_{\>^X}$ are only to the women above and including $\mu(m)$ in $\>_m$ as well as the women in $\{w_a, w_b, w_c, \dots \}$, all of whom he proposes to in $P_{\>^a} \cup P_{\>^b} \cup P_{\>^c} \cup \dots$. Thus, $m_1 \neq m$, implying that $m_1$ is a truthful agent.

Since men propose in decreasing order of their preference, $m_1$ must have been rejected by a woman, say $w_2$, whom he ranks immediately above $w_1$ in $\>_{m_1}$, before proposing to $w_1$ under $\>^X$. Further, $m_1$ must propose to $w_2$ under at least one of the profiles $\>^a,\>^b,\>^c,\dots$, and must not be rejected by her (otherwise, he will propose to $w_1$). Thus, $m_1$ must be matched with $w_2$ under at least one of the matchings $\mu^a,\mu^b,\mu^c,\dots$. Without loss of generality, let us assume that $\mu^a(m_1) =w_2$.

Since $w_2$ is not matched with $m_1$ under $\mu^X$, she must have received a more preferable proposal from some man, say $m_2$; thus, $m_2 \>_{w_2} m_1$. Due to our assumption that $(m_1, w_1)$ is the first proposal during the \DA{} execution on $\>^X$ that does not occur in $P_{\>^a} \cup P_{\>^b} \cup P_{\>^c} \cup \dots$, $m_2$ must have proposed to $w_2$ under at least one of the profiles $\>^a, \>^b, \>^c, \dots$. Without loss of generality, he proposes under $\>^b$. Since women match with their best proposers, $\mu^b(w_2) \succeq_{w_2} m_2$. This, combined with $m_2 \succ_{w_2} m_1 = \mu^a(w_2)$, we get $\mu^b(w_2) \succ_{w_2} \mu^a(w_2)$.

We infer that $m_1$ does not propose to $w_2$ under $\>^b$, since otherwise $w_2$ (eventually) rejects $m_1$ causing him to propose to $w_1$, which would contradict our assumption that $(m_1, w_1) \notin P_{\>^a} \cup P_{\>^b} \cup P_{\>^c} \cup \dots$. Thus, $\mu^b(m_1) \>_{m_1} w_2 = \mu^a(m_1)$. 

Notice that the profiles $\>^a$ and $\>^b$ are obtained from the true preference profile $\>$ by no-regret push up operations. Therefore, from \Cref{lem:SingleAgentPushUpNoRegret}, we get that the matchings $\mu^a$ and $\mu^b$ are stable with respect to the true preference profile (i.e., $\mu^a,\mu^b \in S_{\>}$). 

Now consider the \emph{join} $\mu_{\vee} \coloneqq \mu^a \vee \mu^b$ of the two matchings with respect to the true preference profile $\>$, wherein each man is associated with his more preferred partner between $\mu^a$ and $\mu^b$, and each woman is associated with her less preferred partner (refer to \Cref{subsec:Additional_Preliminaries} for the relevant definitions). Thus, $m_1$ is associated with $\mu^b(m_1) \neq w_2$ and $w_2$ is associated with $\mu^a(w_1) = m_1$. The resulting assignment $\mu_{\vee}$ is not a valid matching, which contradicts the stable lattice result (\Cref{prop:Lattice}).
\end{proof}

We are now ready to prove \Cref{lem:PushUpOneWoman}.

\PushUpOneWoman*

\begin{proof}
Let $m_X \coloneqq \mu^X(w)$, where $\mu^X \coloneqq \DA(\>^X)$. From \Cref{lem:SingleAgentPushUpNoRegret}, we know that the accomplice $m$ does not incur regret under any of the matchings $\mu^a, \mu^b, \mu^c, \dots$, and therefore $\mu^a(m) = \mu^b(m) = \mu^c(m) = \dots = \mu(m)$. If $m_X = m$, then the lemma follows trivially since $m$ is matched with his $\mu^X$-partner, namely $w$, under each of the matchings $\mu^a, \mu^b, \mu^c, \dots$, and therefore the $\mu^X$-partner of $w$ can also be achieved under any of the profiles $\>^a, \>^b, \>^c, \dots$. Thus, for the remainder of the proof, we will assume that $m_X \neq m$; in other words, $m_X$ is a truthful agent.

Suppose, for contradiction, that $m_X$ is not matched to $w$ under any of the profiles $\>^a, \>^b, \>^c, \dots$. Starting from the profile $\>^a$, we can obtain the profile $\>^X$ via a no-regret push up operation of the set $X \setminus \{w_a\}$ in the list $\>_m^{a}$ of the accomplice. Therefore, from \Cref{cor:PushUp}, we have that $\mu^a(m_X) \succeq_{m_X} \mu^X(m_X) = w$, where ${\mu^a \coloneqq \DA(\>^a)}$. Since $\mu^a(m_X) \neq w$ by the contradiction assumption, we further obtain that $\mu^a(m_X) \>_{m_X} w$. By a similar argument, we get that the women $\mu^b(m_X), \mu^c(m_X), \dots$ are also placed above $w$ in the list $\>_{m_X}$.

Since $m_X$ is matched with $w$ under $\mu^X$, he must propose to $w$ during the execution of \DA{} algorithm on $\>^X$, i.e., $(m_X,w) \in P_{\>^X}$. From \Cref{lem:PushUpPropsalsContainment}, we have that ${(m_X,w) \in P_{\>^a} \cup P_{\>^b} \cup P_{\>^c} \dots}$. Since $m_X$ is a truthful agent, his preference list remains unchanged, and therefore he ranks the women $\mu^a(m), \mu^b(m_X), \mu^c(m_X), \dots$ strictly above $w$ under each of the profiles $\>^a, \>^b, \>^c, \dots$. This, however, poses a contradiction since men propose in decreasing order of their preference.
\end{proof}

\subsection{Proof of Corollary~\ref{cor:Optimal_No-regret_PolynomialTime}}
\PolynomialTimeNoRegret*
\begin{proof}(sketch)
The algorithm simply promotes each woman that is below $\mu(m)$ in the accomplice's true preference list to some position above $\mu(m)$ and checks the \DA{} outcome. The total number of such checks is $\O(n)$, and for each check, running the \DA{} algorithm takes $\O(n^2)$ time.
\end{proof}

\subsection{Proof of Corollary~\ref{cor:No-regret_Inconspicuous_Stable}}
\InconspicuousStable*
\begin{proof}(sketch)
An inconspicuous manipulation is a special case of a push up operation, which was shown in \Cref{thm:No-regret_StableLatticeContainment} to be stability preserving.
\end{proof}

\section{Omitted Material from Section~\ref{sec:with-regret}}

\subsection{Proof of Theorem~\ref{thm:Inconspicuous_WithRegret}}\label{subsec:With_Regret_Inconspicuous_Proof}

\InconspicuousWithRegret*

Recall from \Cref{lem:PushUpOneWoman} that the match for the manipulating woman $w$ obtained by a no-regret push up operation of a set $X \subseteq W$ by the accomplice can also be achieved by pushing up exactly one woman in $X$. The following result (\Cref{lem:PushUpOneWomanWithRegret}) establishes the with-regret analogue of this result.

\PushUpOneWomanWithRegret*

\paragraph{Preliminaries for the proof of \Cref{lem:PushUpOneWomanWithRegret}:} The relevant notation is similar to that used in the proof of \Cref{lem:PushUpOneWoman}, which we recall below for the sake of completeness.

Let $\>$ be a true preference profile. Given an accomplice $m$ and a set of women $X = \{w_a, w_b, w_c, \dots \}$ such that $\mu(m) \>_m w_x$ for all $w_x \in X$, we define $\>^a_m$ as the list derived from $\>_m$ where $m$ pushes up $w_a$, $\>^b_m$ as the list derived from $\>_m$ where $m$ pushes up $w_b$, etc., and $\>^X_m$ as the list where $m$ pushes up all women in $X$; the corresponding profiles are $\>^a, \>^b, \>^c, \dots, \>^X$. Additionally, let ${\mu^a \coloneqq \DA(\>^a)}$, ${\mu^b \coloneqq \DA(\>^b)}$, ${\mu^c \coloneqq \DA(\>^c)}$, $\dots$, ${\mu^X \coloneqq \DA(\>^X)}$.

We will now show that under a with-regret push up operation, the \DA{} algorithm matches the  accomplice to one of the women he pushes up (\Cref{lem:accomplice-match-after-regret-manipulation}). In particular, if the accomplice promotes only one woman and incurs regret, then he must be matched to her in the resulting matching.

\begin{restatable}{lemma}{AccompliceMatchAfterWithRegretManipulation}
Let $\>$ be a true preference profile and ${\mu \coloneqq \DA(\>)}$. For any fixed man $m$ and any subset $X \subseteq W$ of women, let $\>' \coloneqq \{ \>_{-m}, \>_m^{X \uparrow} \}$ denote the preference profile after pushing up the women in $X$ in $\>_m$ and let ${\mu' \coloneqq \DA(\>')}$. If $m$ incurs regret (i.e., if $\mu(m) \>_{m} \mu'(m)$), then $\mu'(m) \in X$.
\label{lem:accomplice-match-after-regret-manipulation}
\end{restatable}

\begin{proof}
Suppose, for contradiction, that $\mu'(m) \notin X$. Because of \Cref{prop:Permuting_Falsified_Lists}, without loss of generality, we have that the set $X$ is placed immediately below $\mu(m)$ in the true list $\>_m$. By strategyproofness of the \DA{} algorithm, $m$ cannot be matched under $\mu'$ to any woman  who, according to his true list, is preferred over $\mu(m)$, i.e., $\mu'(m) \notin \{z \in W : z \>_m \mu(m)\}$. By the contradiction assumption, $m$ cannot be matched with any woman in $X$, and because of the with-regret assumption, $m$ cannot be matched with $\mu(m)$ either. Therefore, the woman $y \coloneqq \mu'(m)$ must be such that $\mu(m) \>'_m y$, which, by the push up assumption, implies that $\mu(m) \>_m y$.

Once again, due to \Cref{prop:Permuting_Falsified_Lists}. we can assume, without loss of generality, that $y$ is placed immediately below the set $X$ in the true list $\>_m$. This, in turn, implies that $y$ is immediately below $\mu(m)$ in the manipulated list $\>_m^{X \uparrow}$. Therefore, starting with the list $\>_m^{X \uparrow}$, one can obtain the true list $\>_m$ simply by permuting the agents that are above $y$. \Cref{prop:Permuting_Falsified_Lists} would then imply that the partner of $m$ does not change in the process, i.e., $\mu'(m) = \mu(m)$, which is a contradiction. Hence, we must have $\mu'(m) \in X$.
\end{proof}

\Cref{lem:accomplice-match-after-regret-manipulation} implies that the accomplice $m$ matches with some woman in $X$ under the matching $\mu^X$; say $\mu^X(m) = w_a$. We will now show that $m$ is matched with $w_a$ under the matching $\mu^a$ as well (\Cref{lem:Single_With_Regret_Profile_Match}). Notice that in light of \Cref{prop:Permuting_Falsified_Lists}, we can assume, without loss of generality, that in the lists $\>^X_m$ and $\>^a_m$, the woman $w_a$ is placed immediately above $\mu(m)$. That is, $w_a$ is the least-preferred woman in $X$ according to the list $\>^X_m$.

\begin{restatable}{lemma}{Single_With_Regret_Profile_Match}
Let $w_a \in X$ be the woman who the accomplice $m$ matches with under $\mu^X$. Then, $m$ also matches with $w_a$ under $\mu^a$, where $\mu^a \coloneqq \DA(\>^a)$ and $\>^a \coloneqq \{ \>_{-m}, \>_m^{w_a \uparrow} \}$.
\label{lem:Single_With_Regret_Profile_Match}
\end{restatable}

\begin{proof}
Starting with the profile $\>^X$, we can obtain $\>^a$ by pushing down all women in $X \setminus \{w_a\}$ (recall from the aforementioned observation that $w_a$ is the least-preferred woman in $X$ according to the list $\>^X_m$). Then, from \Cref{prop:PushDown}, we get that $\mu^a(m) = \mu^X(m)$.
\end{proof}

From \Cref{lem:Single_With_Regret_Profile_Match}, it follows that $\>^a$ is a with-regret profile. By contrast, \Cref{lem:No_Regret_Profiles} will show that the rest of the profiles $\>^b, \>^c \dots$ do not cause regret for the accomplice.

\begin{restatable}{lemma}{No_Regret_Profiles}
Let $w_a \in X$ be the woman who the accomplice $m$ matches with under $\mu^X$. Then, for any $w_z\in X \setminus \{w_a\}$, $m$ does not incur regret under the profile $\>^z \coloneqq \{ \>_{-m}, \>_m^{w_z \uparrow} \}$.
\label{lem:No_Regret_Profiles}
\end{restatable}

\begin{proof}
Let $\tilde{X} \coloneqq X \setminus \{w_a\}$, and let $\>^{\tilde{X}} \coloneqq \{ \>_{-m}, \>_m^{\bar{X} \uparrow} \}$ be the profile where $m$ pushes up all women in $\tilde{X}$ starting from the true list $\>_m$. 

We claim that $m$ must match with the woman $\mu(m)$ under the matching $\mu^{\tilde{X}} \coloneqq \DA(\>^{\tilde{X}})$. Indeed, if that is not the case, then promoting $\tilde{X}$ is a with-regret push up operation. Then, from \Cref{lem:accomplice-match-after-regret-manipulation}, the man $m$ must be matched with some woman in $\tilde{X}$ under $\mu^{\tilde{X}}$. This, however, would contradict the strategyproofness of \DA{} algorithm, as $m$ is able to strictly improve in going from the ``true'' list $\>^X_m$ (where he is matched with $w_a$, who is the least-preferred woman in $X$ according to the list $\>^X_m$) to the ``manipulated'' list $\>^{\tilde{X}}$ (where his partner is some woman in $\tilde{X}$).

Thus, $m$ must be matched with the woman $\mu(m)$ under $\mu^{\tilde{X}}$, implying that promoting $\tilde{X}$ is a no-regret push up operation. \Cref{lem:SingleAgentPushUpNoRegret} now implies that for every $w_z \in \tilde{X}$, $\>^z \coloneqq \{ \>_{-m}, \>_m^{w_z \uparrow} \}$ is also a no-regret profile, as desired.
\end{proof}

It can also be shown that if the accomplice $m$ pushes up all women in $X \setminus \{w_a\}$ simulataneously, then he does not incur regret (\Cref{lem:XBar_NoRegret}).

\begin{restatable}{lemma}{XBar_NoRerget}
Let $\>^{\bar{X}}$ be the profile obtained by pushing up $\bar{X} \coloneqq X \setminus \{w_a\}$ in the accomplice $m$'s true preference list. Then, $m$ does not incur regret under $\>^{\bar{X}}$ (i.e, $m$ matches with $\mu(m)$ under $\>^{\bar{X}}$). 
\label{lem:XBar_NoRegret}
\end{restatable}

\begin{proof}
Suppose, for contradiction, that $\>^{\bar{X}}$ is a with-regret profile. Then, from \Cref{lem:accomplice-match-after-regret-manipulation,lem:Single_With_Regret_Profile_Match}, we have that for some $w_z \in \bar{X}$, the profile $\>^z \coloneqq \{\>_{-m},\>_m^{w_z \uparrow}\}$ is also with-regret. This, however, contradicts the implication of \Cref{lem:No_Regret_Profiles} that $\>^z$ is no-regret for every $w_z \in \bar{X}$.
\end{proof}

Recall that given any profile $\>$, $P_{\>}$ denotes the set of all proposals that occur in the execution of the \DA{} algorithm on $\>$. Formally, for any man $m_i \in M$ and woman $w_j \in W$, the ordered pair $(m_i,w_j)$ belongs to the set $P_{\>}$ if $m_i$ proposes to $w_j$ during the execution of \DA{} algorithm on the profile $\>$. Our next result (\Cref{lem:PushUpProposals}) shows that the set of proposals that occur under a true preference profile $\>$ is contained in the set of proposals that occur under a profile $\>'$ that is obtained via a no-regret push up operation on $\>$.

\begin{restatable}{lemma}{PushUpProposals}
Let $\>$ be a preference profile and ${ \mu \coloneqq \DA(\>) }$. For any fixed man $m$, let $\>' = \{ \succ_{-m}, \succ_m^{X\uparrow} \}$ and ${ \mu' = \DA(\>') }$ such that $m$ does not incur regret (i.e., $\mu'(m) = \mu(m)$). Then, $P_{\>} \subseteq P_{\>'}$.
\label{lem:PushUpProposals}
\end{restatable}

\begin{proof}
Since $m$ does not incur regret under $\>'$, we have that $\mu \succeq_M \mu'$ (\Cref{cor:PushUp}). Under the \DA{} algorithm, men propose in decreasing order of their preference. Therefore, any proposal made by a truthful man under $\>$ is also made under $\>'$. Furthermore, the push up assumption implies that the accomplice $m$ proposes to the women in $\>^L_m$ (i.e., the woman strictly preferred by $m$ over $\mu(m)$ according to his true list $\>_m$) under $\>^X$ as well.
\end{proof}

Let $P_{\> \setminus \>'} \coloneqq P_{\>} \setminus P_{\>'}$ denote the set of proposals that occur under the profile $\>$ but not under $\>'$. Our next result (\Cref{lem:NoRegretProfileSubset_Helper}) shows that the set $P_{\>^z \setminus \>^X}$, where $\>^z$ is a no-regret profile obtained by pushing up some woman ${ w_z \in X \setminus \{w_a\} }$ (as established in \Cref{lem:No_Regret_Profiles}), is contained in the set $P_{\>}$. 

\begin{restatable}{lemma}{NoRegretProfileSubset_Helper}
For any woman $w_z \in X \setminus \{w_a\}$, let $\>^z$ be the no-regret profile obtained by pushing up $w_z$ (as discussed in \Cref{lem:No_Regret_Profiles}). Then, $P_{\>^z \setminus \>^X} \subseteq P_{\>}$.
\label{lem:NoRegretProfileSubset_Helper}
\end{restatable}

\begin{proof}
We start by showing that any proposal in $P_{\>^z \setminus \>^X}$ must occur after $m$ proposes to $\mu(m)$ during the \DA{} execution on $\>^z$. Suppose, for contradiction, that this is not true. Then, let $(m_1, w_1)$ be the \emph{first} proposal during the \DA{} execution on $\>^z$ such that it does not belong to $P_{\>^X}$ and occurs before $(m, \mu(m))$. Note that the proposals made by the accomplice $m$ before $\mu(m)$ under $\>^z$ are only to the women above and including $\mu(m)$ in $\>_m$ as well as $w_z$, all of whom he also proposes to under $\>^z$. Thus, $m_1 \neq m$, implying that $m_1$ is a truthful agent.

Since men propose in decreasing order of their preference and it is assumed that $(m_1, w_1) \notin P_{\>^X}$, $m_1$ must have been rejected by $w_2 \coloneqq \mu^X(m_1)$ under $\>^z$ before proposing to $w_1$. Then, under $\>^z$, $w_2$ must have received a proposal from some man, say $m_2$, such that $m_2 \>_{w_2} m_1$. We assumed $(m_1,w_1)$ to be the first proposal during the $\DA$ execution on $\>^z$ to not belong to $P_{\>^X}$. Since $(m_2,w_2)$ occurs before $(m_1,w_1)$ under $\>^z$, we must have that $(m_2,w_2)$ also occurs under $\>^X$. We already established that $m_2 \>_{w_2} m_1$ and know that women are truthful. Since women match with their favorite proposers, this implies that $w_2$ does not match with $m_1$ under $\>^X$. However, this contradicts the statement that $w_2 = \mu^X(m_1)$.

We proceed by showing that any proposal that occurs after $m$ proposes to $\mu(m)$ during the \DA{} execution on $\>^z$ must also occur in $P_{\>}$. Suppose, for contradiction, that this is not true. Then, let $(m_1, w_1)$ be the \emph{last} proposal during the \DA{} execution on $\>^z$ such that it does not belong to $P_{\>}$ and occurs after $(m, \mu(m))$. We infer that $w_1$ accepts and matches with $m_1$ under $\>^z$, since otherwise $m_1$ must make additional proposals, contradicting our assumption that $(m_1, w_1)$ is the last proposal to not occur in $P_{\>}$. Note that $m$ does not propose to anyone after $\mu(m)$ by the no-regret assumption of $\>^z$. Thus, $m_1 \neq m$, implying that $m_1$ is a truthful agent.

Let $m_2$ be the man $w_1$ is temporarily engaged to before she receives a proposal from $m_1$ under $\>^X$. We assumed that $w_2 = \mu^X(m_1)$ which means that $w_1$ rejects $m_2$ in favor of $m_1$ under $\>^z$, causing $m_2$ to propose to the next woman in $\>^z_{m_2}$, say $w_2$. Since $P_{\>} \subseteq P_{\>^z}$ (\Cref{lem:PushUpProposals}), we know that $(\mu(w_1), w_1) \in P_{\>^z}$. We also defined $m_2$ to be $w_1$'s favorite proposer under $\>^z$ before matching with $m_1$, which implies that $m_2 \succeq_{w_1} \mu(w_1)$. If $w_1 \>_{m_2} \mu(m_2)$, then the pair $(m_2, w_1)$ blocks $\mu$ with respect to $\>$. Thus, we have that $\mu(m_2) \succeq_{m_2} w_1$. Note that the accomplice $m$ does not make any proposals after $(m_1, w_1)$ under $\>^z$, since  $(m_1, w_1)$ occurs after $(m, \mu(m))$ and $\>^z$ is a no-regret profile. We know that $m_2$ proposes to $w_2$ after $(m_1, w_1)$, implying that $m_2 \neq m$. Therefore, $\mu(m_2) \succeq_{m_2} w_1 \>_{m_2} w_2$, and $m_2$ does not propose to $w_2$ under $\>$. However, this contradicts the assumption that $(m_1, w_1)$ is the last proposal during the \DA{} execution on $\>^X$ to not belong to $P_{\>}$. 

Thus, we have shown that (1) any proposal that belongs to $P_{\>^z \setminus \>^X}$ occurs after $(m, \mu(m))$ during the \DA{} execution on $\>^z$, and (2) any proposal that occurs after $(m, \mu(m))$ during the \DA{} execution on $\>^z$ belongs to $P_{\>}$. These two statements imply that $P_{\>^z \setminus \>^X} \subseteq P_{\>}$.
\end{proof}

Our next result (\Cref{lem:PushUpPropsalsContainment_Regret_Helper}) shows that the set $P_{\>^X \setminus \>^{\bar{X}}}$ is contained in the set $P_{\>^a}$. 

\begin{restatable}{lemma}{PushUpProposalsContainment_Regret_Helper}
Let $\>^{\bar{X}}$ be the profile obtained by pushing up $\bar{X} \coloneqq X \setminus \{w_a\}$ in the accomplice $m$'s true preference list. Then, $P_{\>^X \setminus \>^{\bar{X}}} \subseteq P_{\>^a}$.
\label{lem:PushUpPropsalsContainment_Regret_Helper}
\end{restatable}

\begin{proof}
For this proof, let $\>^a$ be the ``true'' preference profile. Remember, from \Cref{lem:XBar_NoRegret}, that $m$ matches with $\mu(m)$ under $\>^{\bar{X}}$.  Thus, $\>^{\bar{X}}$ is a with-regret profile with respect to $\>^a$, and is obtained by pushing up $\hat{X} \coloneqq \bar{X} \cup \{\mu(m)\}$ in $\>^a_m$. On the other hand, $m$ matches with $w_a$ under $\>^a$ and $\>^X$ (\Cref{lem:Single_With_Regret_Profile_Match}). Thus, $\>^X$ is a no-regret profile with respect to $\>^a$, and is obtained by pushing up $\bar{X} = \hat{X} \setminus \{\mu(m)\} = X \setminus \{w_a\}$.

For any woman $w_z \in \bar{X}$, let $\>^{az}$ be the profile obtained by pushing up $w_z$ in $\>^a_m$. It is easy to see that $\>^{az}$ is a no-regret profile with respect to $\>^a$. Indeed, suppose $\>^{az}$ is with-regret. Then, $m$ must match with $w_z$ under $\>^{az}$ (\Cref{lem:accomplice-match-after-regret-manipulation}). In light of \Cref{prop:Permuting_Falsified_Lists} we can assume, without loss of generality, that $m$ ranks the set $\bar{X}$ strictly above $w_a$ in the list $\>^X_m$. Given profile $\>^X$, $m$ could then manipulate by submitting $\>^{az}_m$ in order to match with $w_z$. However, this would contradict strategyproofness of the \DA{} algorithm since $w_z \>_m w_a$. Therefore, $\>^{az}$ is no-regret with respect to $\>^a$.

Given this observation, from \Cref{lem:NoRegretProfileSubset_Helper}, we get that for any woman $w_z \in \bar{X} = \hat{X} \setminus \{\mu(m)\}$, $P_{\>^{az}} \setminus P_{\>^{\bar{X}}} \subseteq P_{\>^a}$. Since $m$ matches with $w_a$ under $\>^{az}$ and $\>^X$, we have that $\>^{az}$ is no-regret with respect to $\>^X$. Thus, we invoke \Cref{lem:PushUpPropsalsContainment} to claim that any proposal that occurs under $\>^X$ is contained in $P_{\>^{az_1}} \cup P_{\>^{az_2}} \cup \dots \cup P_{\>^{az_k}}$, where $\bar{X} \coloneqq \{w_{z_1}, w_{z_2}, \dots, w_{z_k}\}$. Since $\{P_{\>^{az_1}} \cup P_{\>^{az_2}} \cup \dots \cup P_{\>^{az_k}}\} \setminus P_{\>^{\bar{X}}} \subseteq P_{\>^a}$, we get that $P_{\>^X \setminus \>^{\bar{X}}} \subseteq P_{\>^a}$.
\end{proof}

Recall from \Cref{lem:PushUpPropsalsContainment} that, in the no-regret setting, any proposal that occurs during the \DA{} execution on the profile $\>^X$ must also occur during the \DA{} execution on at least one of the profiles $\>^a, \>^b, \>^c, \dots$. Our next result (\Cref{lem:PushUpPropsalsContainment_Regret}) establishes the with-regret analogue of this result.

\begin{restatable}{lemma}{PushUpProposalsContainment_Regret}
Let $X = \{w_a,w_b,w_c,\dots\}$ be an arbitrary finite set of women such that the accomplice $m$ incurs regret after pushing up $X$. Then, any proposal that occurs under $\>^X$ also occurs under at least one of the profiles $\>^{a}, \>^{b}, \>^{c}, \dots$.
\label{lem:PushUpPropsalsContainment_Regret}
\end{restatable}

\begin{proof}
Let $\>^{\bar{X}}$ be the profile after $m$ pushes up $\bar{X} \coloneqq X \setminus \{w_a\}$ in his true preference list. Then, from \Cref{lem:PushUpPropsalsContainment_Regret_Helper}, we know that $P_{\>^X} \setminus P_{\>^{\bar{X}}} \subseteq P_{\>^a}$. Additionally, from \Cref{lem:XBar_NoRegret}, we know that $\>^{\bar{X}}$ is a no-regret profile. It then follows from  \Cref{lem:PushUpPropsalsContainment} that $P_{\>^{\bar{X}}} \subseteq P_{\>^b} \cup P_{\>^c} \cup \dots$.

Any proposal that occurs under $\>^X$ is contained in either $P_{\>^X} \setminus P_{\>^{\bar{X}}}$ or $P_{\>^{\bar{X}}}$. By combining the aforementioned observations, we get that any such proposal is contained in $P_{\>^a} \cup P_{\>^b} \cup P_{\>^c} \cup \dots$, as desired.
\end{proof}

We are now ready to prove \Cref{lem:PushUpOneWomanWithRegret}.

\PushUpOneWomanWithRegret*

\begin{proof}
Let $m_X \coloneqq \mu^X(w)$, where $\mu^X \coloneqq \DA(\>^X)$. We assume that the accomplice $m$ matches with a woman $w_a \in X$ (\Cref{lem:accomplice-match-after-regret-manipulation}) under $\>^X$. From \Cref{lem:Single_With_Regret_Profile_Match}, we know that he also matches with $w_a$ under $\>^a$. If $m_X = m$, then the lemma follows trivially since $m$ is matched with his $\mu^X$-partner, namely $w$, under the matching $\mu^a$, and therefore the $\mu^X$-partner of $w$ can also be achieved through profile $\>^a$. Thus, for the remainder of the proof, we assume that $m_X \neq m$; in other words, $m_X$ is a truthful agent.

Suppose, for contradiction, that $m_X$ is not matched to $w$ under any of the profiles $\>^a, \>^b, \>^c, \dots$. Starting from the with-regret profile $\>^a$ (as established in \Cref{lem:Single_With_Regret_Profile_Match}), we can obtain the profile $\>^X$ via a no-regret push up operation of the set $X \setminus \{w_a\}$ in the list $\>_m^{a}$ of the accomplice. Therefore, from \Cref{cor:PushUp}, we have that $\mu^a(m_X) \succeq_{m_X} \mu^X(m_X) = w$, where ${\mu^a \coloneqq \DA(\>^a)}$. Since $\mu^a(m_X) \neq w$ by the contradiction assumption, we further obtain that $\mu^a(m_X) \>_{m_X} w$. 

Now, consider the no-regret profiles $\>^b, \>^c, \>^d, \dots$ (as established in \Cref{lem:No_Regret_Profiles}). Suppose $\mu^X(m_X) \>_{m_X} \mu^b(m_X)$, where $\mu^b \coloneqq \DA(\>^b)$. Then, since men propose in decreasing order of their preference, $(m_X, \mu^b(m_X)) \in P_{\>^b \setminus \>^X}$. \Cref{lem:NoRegretProfileSubset_Helper} consequently implies that $(m_X, \mu^b(m_X)) \in P_{\>}$, and thus $\mu^b(m_X) \succeq_{m_X} \mu(m_X)$. Since $\mu(m_X) \succeq_{m_X} \mu^b(m_X)$ by \Cref{cor:PushUp}, we infer that $\mu^b(m_X) = \mu(m_X)$. This, combined with $\mu^X(m_X) \>_{m_X} \mu^b(m_X)$, gets us $w = \mu^X(m_X) \>_{m_X} \mu(m_X)$. From the accomplice manipulation assumption, we also know that $m_X = \mu^X(w) \>_{w} \mu(w)$. However, this implies that the pair $(m_X, w)$ blocks $\mu$ with respect to $\>$, posing a contradiction. Thus, we have shown that $w = \mu^X(m_X) \nsucc_{m_X} \mu^b(m_X)$. Since $\mu^b(m_X) \neq w$ by the initial contradiction assumption, we further obtain that $\mu^b(m_X) \>_{m_X} w$. By a similar argument, we get that the women $\mu^c(m_X), \mu^d(m_X), \dots$ are also placed above $w$ in the list $\>_{m_X}$.

Since $m_X$ is matched with $w$ under $\mu^X$, he must propose to $w$ during the execution of \DA{} algorithm on $\>^X$, i.e., $(m_X,w) \in P_{\>^X}$. From \Cref{lem:PushUpPropsalsContainment_Regret}, we have that ${(m_X,w) \in P_{\>^a} \cup P_{\>^b} \cup P_{\>^c} \dots}$. Since $m_X$ is a truthful agent, his preference list remains unchanged, and therefore he ranks the women $\mu^a(m), \mu^b(m_X), \mu^c(m_X), \dots$ strictly above $w$ under each of the profiles $\>^a, \>^b, \>^c, \dots$. This, however, poses a contradiction since men propose in decreasing order of their preference.
\end{proof}

We are now ready to prove \Cref{thm:Inconspicuous_WithRegret}.

\InconspicuousWithRegret*
\begin{proof}
From \Cref{prop:Permuting_Falsified_Lists} (and subsequent remarks), we know that any accomplice manipulation can be simulated via push up and push down operations. \Cref{lem:CombiningPushUpPushDown} shows that any beneficial manipulation that is achieved by some combination of pushing up a set $X \subseteq W$ and pushing down $Y \subseteq W$ can be weakly improved by only pushing up $X \subseteq W$.
Finally, from \Cref{lem:PushUpOneWomanWithRegret}, we know that any match for the manipulating woman $w$ that is achieved by pushing up $X \subseteq W$ is also achieved by pushing up exactly one woman in $X$, thus establishing the desired inconspicuousness property.
\end{proof}

\subsection{Proof of Corollary~\ref{cor:Optimal_With-regret_PolynomialTime}}

\PolynomialTimeWithRegret*

\begin{proof}(sketch)
The algorithm simply promotes each woman that is below $\mu(m)$ in the accomplice's true preference list to some position above $\mu(m)$ and checks the \DA{} outcome. The total number of such checks is $\O(n)$, and for each check, running the \DA{} algorithm takes $\O(n^2)$ time.
\end{proof}

\section{Omitted Material from Section~\ref{sec:Experiments}}

Recall from our previous experiments that we generated 1000 profiles uniformly at random for each value of $n \in {3, . . . , 40}$ (where $n$ is the number of men/women) and allowed any man to be chosen as an accomplice. We follow the same setup for all subsequent experiments unless stated otherwise. 

\subsection{Fraction of Women Who Improve (Cont'd)}

We revisit the experiment in which we compared the fraction of women who are able to improve through no-regret accomplice and self manipulation. In \Cref{tab:WomenWhoImprove}, we catalog the number of women who benefit from both strategies when $n = 20$.

\begin{table*}[ht]
\centering
\resizebox{\textwidth}{!}{
\begin{tabular}{lcccccccccccccccccccccc} 
    No. of women who benefit & 0 & 1 & 2 & 3 & 4 & 5 & 6 & 7 & 8 & 9 & 10 & 11 & 12 & 13 & 14 & 15 & 16 & 17 & 18 & 19 & 20\\
    \hline
    No. of instances (accomplice) & 307 & 0 & 101 & 102 & 88 & 79 & 88 & 59 & 49 & 54 & 28 & 21 & 11 & 8 & 3 & 1 & 1 & 0 & 0 & 0 & 0\\
    No. of instances (self) & 411 & 151 & 178 & 128 & 68 & 33 & 20 & 7 & 1 & 2 & 1 & 0 & 0 & 0 & 0 & 0 & 0 & 0 & 0 & 0 & 0
\end{tabular}}
\caption{The number of instances (out of 1000) where a varying numbers of women benefit through no-regret accomplice manipulation and self manipulation when $n$ = 20.}
\label{tab:WomenWhoImprove}
\end{table*}

Not only are there more instances where at least one woman improves through no-regret accomplice manipulation, but there are also more instances where larger numbers of women improve. For example, there are no instances where more than ten women improve through self manipulation. This is a stark contrast to no-regret accomplice manipulation through which sixteen women are able to improve in one of the instances. Interestingly, there are no instances where \emph{exactly one} woman improves through no-regret accomplice manipulation. This is due to the sequence of proposals that occur after a no-regret push up operation. We formalize this observation in \Cref{prop:StrictPushUp}.

\begin{restatable}{proposition}{StrictPushUp}
Let $\succ$ be a preference profile and $\mu \coloneqq \DA(\succ)$. For any man $m$, let $\succ' \coloneqq \{\succ_{-m}, \succ_m^{X\uparrow}\}$ and $\mu'\coloneqq \DA(\succ')$. If $m$ does not incur regret and $\mu' \neq \mu$, then there exist at least two distinct women $w',w'' \in W$ such that $\mu'(w') \>_{w'} \mu(w')$ and $\mu'(w'') \>_{w''} \mu(w'')$, and at least two distinct men $m',m'' \in M$ such that $\mu(m') \>_{m'} \mu'(m')$ and $\mu(m'') \>_{m''} \mu'(m'')$.
\label{prop:StrictPushUp}
\end{restatable}

Before proving \Cref{prop:StrictPushUp}, we show that if a man $m$ performs a no-regret push up operation such that all the women he pushes up prefer him less than their original \DA{} partners, then the push up operation is \emph{weak} (i.e., the manipulated matching is the same as the original matching). 

\begin{restatable}{lemma}{WeakPushUp}
Let $\succ$ be a preference profile and $\mu \coloneqq \DA(\>)$. For any man $m$, let $\>' = \{ \succ_{-m}, \succ_m^{X\uparrow} \}$ and $\mu' = \DA(\>')$ such that $m$ does not incur regret. If $\mu(w') \>_{w'} m$ for all $w' \in X$, then $\mu' = \mu$.
\label{lem:WeakPushUp}
\end{restatable}

\begin{proof}
Since $m$ does not incur regret after a push up operation, $\mu'$ must be stable with respect to the true preferences, i.e., $\mu' \in S_{\>}$ (\Cref{thm:No-regret_StableLatticeContainment}). Additionally, \Cref{cor:PushUp} implies that $\mu'$ is weakly better for all women (i.e., $\mu' \succeq_W \mu$) and weakly worse for all men (i.e., $\mu \succeq_M \mu'$).

Suppose, for contradiction, that $\mu' \neq \mu$. Let $Z \coloneqq \{z \in W: \mu'(z) \>_z \mu(z)\}$ denote the set of women with a strictly more preferable partner under $\mu'$. Thus, by assumption, $Z \neq \emptyset$. Let $Y \coloneqq \{y \in M: \mu'(y) \in Z\}$ denote the set of men whose $\mu'$-partners are in the set $Z$. Note that any man not in $Y$ has the same partner under $\mu$ and $\mu'$; in particular, $m \notin Y$ by the no-regret assumption. Also note that each man in $Y$ strictly prefers its partner under $\mu$ than under $\mu'$, i.e., $\mu \>_Y \mu'$ (this is an easy consequence of the stability of $\mu'$ with respect to the true profile $\>).$

Consider the execution of the \DA{} algorithm on the profile $\>'$. Since the men propose in decreasing order of their preference, each man in $Y$ must be rejected by his $\mu$-partner during the algorithm. Let $m_1 \in Y$ denote the man who is the \emph{earliest} to be rejected by his $\mu$-partner (i.e., the woman he is matched to under the matching $\mu$), say $w_1 \coloneqq \mu(m_1)$. Then, $w_1$ must have at hand a more preferable proposal, say $m_2$ (i.e., $m_2 \>_{w_1} m_1$), that she does not receive under the execution of $\DA(\>)$.

Since $m_2$ is not rejected by his $\mu$-partner yet, we have that $w_1 \>'_{m_2} \mu(m_2)$. Now, if $m_2 \neq m$, then we get $w_1 \>_{m_2} \mu(m_2)$, which contradicts the stability of $\mu$ with respect to $\>$ as $(m_2,w_2)$ would constitute a blocking pair. Otherwise, if $m_2 = m$, then it must be that $w_1 \in X$ (since the women in $X$ are the only ones who $m$ proposes to under $\>'$ but not under $\>$). Then, from the above condition, we will have $m = m_2 \>_{w'} m_1 = \mu(w')$ for some $w' \in X$, which contradicts the condition given in the lemma. Hence, $\mu=\mu'$, as desired.
\end{proof}

We are now ready to prove \Cref{prop:StrictPushUp}.

\begin{proof}[Proof of \Cref{prop:StrictPushUp}]
Let $X \subseteq W$ be the set of women that $m$ pushes up. Since it is assumed that $\mu' \neq \mu$, the contrapositive of \Cref{lem:WeakPushUp} implies there exists a woman $w' \in X$ such that $\mu(w') \nsucc_{w'} m$. By the push up assumption, $\mu(m) \neq w'$; otherwise $m$ would not have been able to push up $w'$. This implies that $m \neq \mu(w')$ and thus, $m \>_{w'} \mu(w')$. Since men propose in decreasing order, $m$ proposes to $w'$ under $\DA{(\>')}$. By no-regret assumption $\mu(m) = \mu'(m)$, we know that $w'$ rejects $m$ at some point to be matched with another man $m'$ such that $m' \>_{w'} \mu(w')$.
This implies that $m'$ did not propose to $w'$ under $\DA{(\>)}$, and thus, $m'$ is matched to a more preferred woman under $\>$, i.e., $\mu(m') \>_{m'} \mu'(m')$.

Now let $\hat{m} \coloneqq \mu(w')$ be $w'$'s partner under $\>$. By \Cref{cor:PushUp}, it must be the case that $\mu(\hat{m}) \>_{\hat{m}} \mu'(\hat{m})$. Let $w'' \coloneqq \mu'(\hat{m})$. Following the same reasoning as above, since $\mu(m) = \mu'(m)$ and $m \neq \mu(w'')$, we have $m \>_{w''} \mu(w'')$. The order of proposals under $\>'$ indicates that $w''$ rejects $m$'s proposal to be matched to a more preferred man $m''$ under $\DA{(\>')}$, which consequently implies that $\mu'(w'') \>_{w''} \mu(w'')$ while $\mu(m'') \>_{m''} \mu(m'')$. Therefore, $w'$ and $w''$'s partners are strictly improved whereas $m'$ and $m''$'s partners are strictly worse off under $\DA{(\>')}$.
\end{proof}

In the same experiment, we additionally computed the fraction of instances in which it is possible for at least one woman to improve through no-regret accomplice and self manipulation. The results in \Cref{fig:ManipulableInstances} once again suggest that no-regret accomplice strategies are more prevalent than self manipulation.

\begin{figure}[!htb]
    \centering
    \begin{minipage}{\textwidth}
        \centering
        \includegraphics[width=0.72\linewidth]{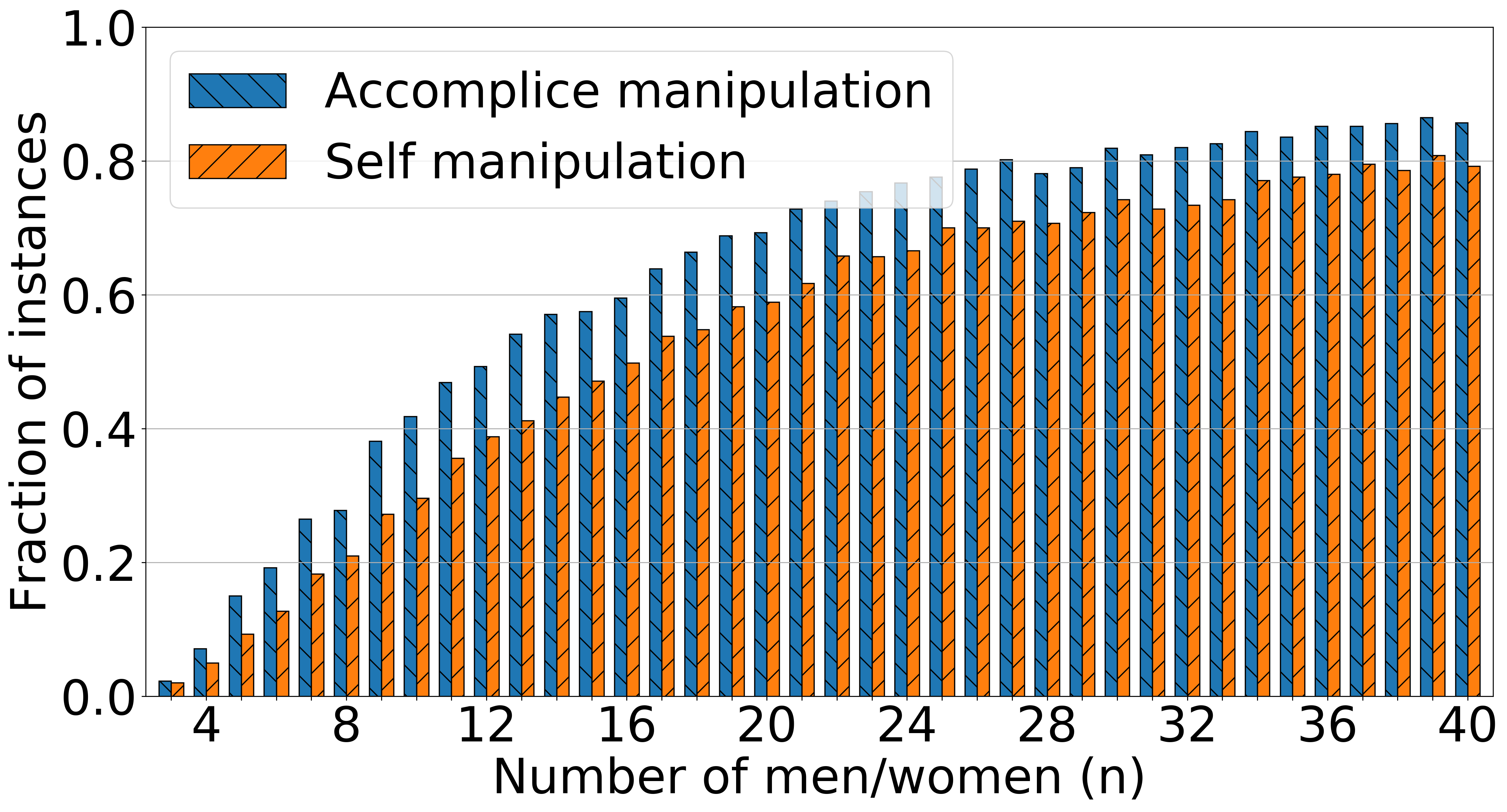}
    \end{minipage}\hfill
    \caption{Comparing fractions of instances that are manipulable by at least one woman through no-regret accomplice and self manipulation.}
    \label{fig:ManipulableInstances}
\end{figure}

\subsection{Regret vs. Improvement}
We examine the tradeoff between regret (of the accomplice) and improvement (of the strategic woman) under the with-regret manipulation model. Rather than allowing any man to be chosen as an accomplice, we ran the with-regret manipulation strategy on a fixed woman $w$ and recorded the outcomes after individually using each man as an accomplice.  In other words, if there were multiple optimal strategies for $w$, we chose the one that caused the accomplice to incur the least amount of regret. \Cref{fig:RegretvImprovement} illustrates the distribution of improvement achieved for $w$ and regret incurred by the accomplices in terms of the difference in the ranks of their matched partners before and after manipulation. Interestingly, the expected regret for the accomplice is greater than the expected improvement for the manipulating woman.

\begin{figure}[t]
    \centering
    \begin{minipage}{\linewidth}
        \centering
        \includegraphics[width=0.72\linewidth]{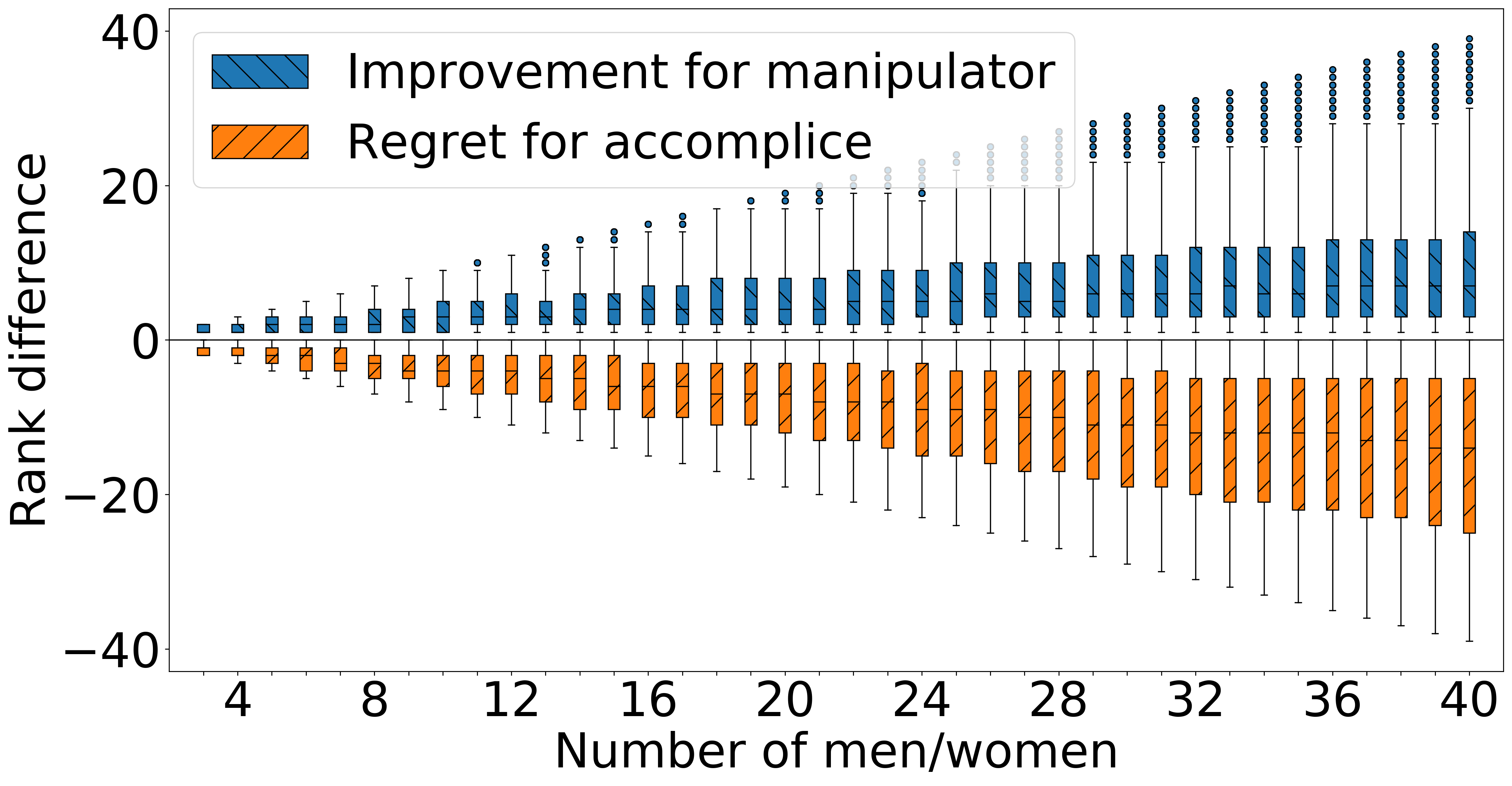}
    \end{minipage}\hfill
    \caption{Comparing distributions of improvement for the strategic woman $w$ and regret for the accomplices. The solid bars, whiskers, and dots denote the interquartile range, range excluding outliers, and outliers, respectively.}
    \label{fig:RegretvImprovement}
\end{figure}

\section{Suboptimal Accomplice Manipulation}

We have already shown that any \emph{optimal} accomplice manipulation strategy admits an equivalent inconspicuous strategy (\Cref{thm:Inconspicuous_NoRegret,thm:Inconspicuous_WithRegret}). \Cref{thm:Suboptimal_Inconspicuous} strengthens this result by showing that \emph{any} beneficial (i.e., optimal or suboptimal) accomplice manipulation admits an equivalent inconspicuous strategy for both the no-regret and with-regret settings. In order to prove this, we start by introducing some new results (\Cref{lem:Push_Up_Subset_Single_Agent,lem:CombiningPushUpPushDown_Suboptimal}).

\begin{restatable}{lemma}{PushUpSubsetSingleAgent}
Let $(m,w)$ be a manipulating pair. Let $X$ be a set of women who $m$ can push up and $\mu(w)$ be $w$’s match after $m$ pushes up all women in $X$. Let $w' \in X$ be the single woman $m$ needs to push up to get $w$ matched with $\mu(w)$ (as guaranteed by \Cref{lem:PushUpOneWoman,lem:PushUpOneWomanWithRegret}). Then, $m$ can push up any subset of women in $X$ that contains $w'$ (i.e., any subset $S \subseteq X$ such that $w' \in S$) to get $w$ matched with $\mu(w)$.
\label{lem:Push_Up_Subset_Single_Agent}
\end{restatable}

\begin{proof}
Let $\>^X \coloneqq \{ \>_{-m}, \>_{m}^{X \uparrow} \}$, $\>^S \coloneqq \{ \>_{-m}, \>_{m}^{S \uparrow} \}$, and $\>^{w'} \coloneqq \{ \>_{-m}, \>_{m}^{w' \uparrow} \}$. Without loss of generality, $\>^S$ is derived from $\>^X$ if $m$ pushes down the set $X \setminus S$. Similarly, $\>^{w'}$ is derived from $\>^S$ if $m$ pushes down the set $S \setminus \{w'\}$. From \Cref{lem:PushDown_Worse_For_Women}, we get that $\mu^X \succeq_W \mu^S$ and $\mu^S \succeq_W \mu^{w'}$, where $\mu^X \coloneqq \DA(\>^X)$, $\mu^S \coloneqq \DA(\>^S)$, and $\mu^{w'} \coloneqq \DA(\>^{w'})$. Since it is assumed that $\mu^X(w) = \mu^{w'}(w)$, it must be the case that $\mu^S(w) = \mu^X(w)$. Thus, $m$ can push up any subset to get $w$ matched with $\mu(w)$. 
\end{proof}

The next result (\Cref{lem:CombiningPushUpPushDown_Suboptimal}) shows that a strictly beneficial accomplice manipulation that uses a combination of push up and push down operations can be achieved through push up operations alone. This strengthens \Cref{lem:CombiningPushUpPushDown} which showed that a combination of push up and push down operations is \emph{weakly worse} than push up operations alone. 

\begin{restatable}{lemma}{CombiningPushUpPushDown_Suboptimal}
Let $(m, w)$ be a manipulating pair, and let $\>$ be a preference profile. For any subsets of women $X \subseteq W$ and $Y \subseteq W$, let $\>' \coloneqq \{ \succ_{-m}, \succ_m^{X\uparrow} \}$ denote the preference profile after pushing up the set $X$ and $\>'' \coloneqq \{ \succ_{-m}, \succ_m^{X\uparrow, Y\downarrow} \}$ denote the profile after pushing up $X$ and pushing down $Y$ in the true preference list $\>_m$ of man $m$. Let $\mu \coloneqq \DA(\>)$, $\mu' \coloneqq \DA(\>')$, and $\mu'' \coloneqq \DA(\>'')$. If $\mu''(w) \succ_w \mu(w)$, then $\mu''(w) = \mu'(w)$.
\label{lem:CombiningPushUpPushDown_Suboptimal}
\end{restatable}

\begin{proof}
Suppose, for contradiction, that $\mu''(w) \neq \mu'(w)$. This, combined with $\mu'(w) \succeq_w \mu''(w)$ (\Cref{lem:CombiningPushUpPushDown}), gets us $\mu'(w) \succ_w \mu''(w)$. Additionally, we have assumed that $\mu''(w) \succ_w \mu(w)$, which gets us $\mu'(w) \succ_w \mu(w)$. 

\begin{figure}[h]
\centering
\begin{tikzcd}[cramped, column sep=tiny]
    \succ_m: & \cdots & Y & \cdots & \mu(m) & \cdots & X & \cdots\\
    \succ'_m: & \cdots & X & \cdots & Y & \cdots & \mu(m) & \cdots\\
    \succ''_m: & \cdots & X & \cdots & \mu(m) & \cdots & Y & \cdots\\
    \succ^*_m: & \cdots & \mu(m) & \cdots & X & \cdots & Y & \cdots
\end{tikzcd}
\caption{An illustration of man $m$'s preference lists under the profiles $\succ$,  $\succ^*$, $\succ'$, and $\succ''$ in the proof of \Cref{lem:CombiningPushUpPushDown_Suboptimal}.}
\label{fig:combining-push-up-push-down-suboptimal}
\end{figure}

Now, let $Z \coloneqq X \cup Y$. Consider a preference profile $\>^*$ derived from the true profile $\>$ by pushing down all women in $Z$ below $\mu(m)$ in the accomplice $m$'s true preference list (see \Cref{fig:combining-push-up-push-down-suboptimal}). From \Cref{prop:PushDown}, we have that $\mu^*(m) = \mu(m)$, where $\mu^* \coloneqq \DA(\>^*)$. This implies that a push up/down operation with respect to $\mu^*(m)$ is equivalent to the same operation with respect to $\mu(m)$. Therefore, starting with $\succ^*_m$, if $m$ pushes up all women in $Z$, then we obtain the profile $\succ'$. Similarly, starting with $\succ^*_m$, if $m$ instead pushes up all women in $X$ (respectively, $Y$), then we obtain the profile $\succ''$ (respectively, $\>$).

Since profile $\>'$ is derived from $\>^*$ via a push up operation of the set $Z$, \Cref{lem:PushUpOneWoman} implies that the same match for $w$, namely  $\mu'(w)$, can be achieved by promoting just one woman, say $w' \in Z$, in the preference list $\>'_m$. Since $X$ and $Y$ are disjoint sets, we have that either $w' \in X$ or $w' \in Y$. If $w' \in X$, then from \Cref{lem:Push_Up_Subset_Single_Agent}, we have that $\mu''(w) = \mu'(w)$, contradicting our original assumption. On the other hand, if $w' \in Y$, then again from \Cref{lem:Push_Up_Subset_Single_Agent} we get $\mu(w) = \mu'(w)$, contradicting the condition $\mu'(w) \succ_w \mu(w)$ that we showed above.
\end{proof}

\begin{restatable}{theorem}{Suboptimal_Inconspicuous}
Any beneficial accomplice manipulation is, without loss of generality, inconspicuous.
\label{thm:Suboptimal_Inconspicuous}
\end{restatable}
\begin{proof}
From \Cref{prop:Permuting_Falsified_Lists} (and the subsequent remarks), we know that any accomplice manipulation can be simulated via push up and push down operations. \Cref{lem:CombiningPushUpPushDown_Suboptimal} shows that \emph{any} beneficial (i.e., optimal or suboptimal) manipulation that is achieved by some combination of pushing up a set $X \subseteq W$ and pushing down $Y \subseteq W$ can also be achieved by only pushing up $X \subseteq W$. Finally, from \Cref{lem:PushUpOneWoman}, we know that any match for the manipulating woman $w$ that is achieved by pushing up $X \subseteq W$ is also achieved by pushing up exactly one woman in $X$, thus establishing the desired inconspicuousness property.
\end{proof}

\end{document}